\documentclass[journal]{IEEEtran}
\usepackage[english]{babel}
\usepackage{amsmath,amssymb,amsfonts}
\usepackage{amsthm}
\usepackage{graphicx}
\usepackage{xcolor}
\usepackage{subcaption}
\usepackage{booktabs}
\usepackage{algorithmic}
\usepackage{algorithm}
\usepackage{hyperref}
\usepackage{comment}

\newtheorem{theorem}{Proposition}

\def\BibTeX{{\rm B\kern-.05em{\sc i\kern-.025em b}\kern-.08em
    T\kern-.1667em\lower.7ex\hbox{E}\kern-.125emX}}
\begin{document}

\title{Orchestrating Key Scarcity: Global Capacity Allocation in Hybrid QKD-PQC Networks}

\author{%
\IEEEauthorblockN{
Pablo P. Rejo$^{1}$,
Manuel Fernández-Veiga$^{1}$,
Ana Fernández-Vilas$^{1}$,\\
Diego Cajaraville-Aboy$^1$,
Javier Vecino-Peñas$^1$,
Rebeca P. Díaz-Redondo$^{1}$
}

\IEEEauthorblockA{
$^{1}$atlanTTic Research Center, Universidade de Vigo, Spain\\
}
\thanks{This work was funded by grant 0590\_QUANTUM\_IBER\_IA\_6\_E. Strategic impulse of capabilities in quantum technologies and artificial intelligence across the Iberian space, co-funded by EU under the programme
Interreg VI-A Spain-Portugal (POCTEP) 2021-2027}%
}
\maketitle

\begin{abstract}
This article presents a Distributed Key Management System (DKMS) that provisions quantum-safe symmetric encryption keys to applications across a network. To overcome the distance and rate limitations inherent to Quantum Key Distribution (QKD), the architecture adopts a hybrid design: it utilizes multi-hop trusted relays for physical key transport, while applying an end-to-end Post-Quantum Cryptography (PQC) layer to ensure the relayed key material remains strictly protected from intermediate nodes. Because QKD transforms key generation into a scarce physical resource, managing unpredictable application demands poses a critical capacity allocation problem. The core contribution of this work is an SDN-driven control plane that globally orchestrates the network. Rather than relying on vulnerable hop-by-hop capacity reservations, a central controller simultaneously computes optimal filling rates and multipath routing for all per-peer key buffers. This global orchestration ensures that no application is starved and that all key buffers synchronize their filling processes. Evaluated on a Kubernetes-based digital twin across five distinct network topologies, the controller successfully filled all key buffers to capacity in under six minutes. Furthermore, under deliberately excessive demand, the network demonstrated graceful degradation, safely rejecting excess requests through controlled rate-limiting rather than experiencing catastrophic failure.
\end{abstract}


\section{Introduction}
\label{sec:intro}

A Key Management System (KMS) is a service whose purpose is to supply cryptographic keys to the applications that need them, exposing an API through which a client can request key material instead of generating and safeguarding it on its own. Its value lies in concentrating behind that single interface two operations that are easy to get wrong in isolation: the generation of the keys, which must draw on a sound entropy source and follow the appropriate cryptographic procedures; and their delivery, which must ensure that each key reaches only its legitimate recipient. A KMS thus lets applications consume keys as a service while assuming responsibility for both halves of one guarantee ---that the keys handed out are well generated and that they are securely distributed to whoever is entitled to them. It is precisely the strength of this guarantee that quantum computing calls into question.

In the emerging landscape of Quantum Key Distribution Networks (QKDN), KMSs  serves as the  control plane that bridges the physical quantum layer ---characterized by variable, environmentally sensitive, and distance-limited key generation rates--- with the dynamic security demands of the application layer. A KMS serves as the trust pillar for secure communications, governing the complete operational lifecycle of cryptographic keys—generation, storage, distribution, rotation, and revocation. In classical environments, a KMS presents a unified interface (such as KMIP or PKCS\#11) for Secure Application Entities (SAEs), agnostic from underlying cryptographic mechanisms. Traditional KMS architectures are based on two assumptions that make this abstraction simple and scalable: (1)  key generation is computationally instantaneous and inexhaustible, i.e. a classical KMS can produce key material on demand whenever a client requests it; and (2) key distribution relies on asymmetric cryptography (such as RSA or Diffie-Hellman), enabling two endpoints to establish shared secrets over untrusted, multi-hop networks without raw key material ever touching intermediate nodes.

However, the quantum scenario undermines both assumptions, as Shor’s algorithm eliminates the security of asymmetric key exchange. In response, two lines of defense have emerged: the information-theoretical security of Quantum Key Distribution (QKD), e.g.~\cite{bb84} and post-quantum cryptography (PQC)~\cite{mceliece1978}. While QKD restores security, it shifts key generation to a rate-limited physical process, turning symmetric keys into a scarce resource and creating a mismatch with unpredictable application demands. Consequently, adapting to this scenario requires Key Management Systems (KMS) to evolve from localized, on-demand generation into a distributed framework driven by per-peer buffer management, network-wide capacity orchestration, and hybrid QKD–PQC transport to extend quantum-safe security to end devices.

A primary challenge in deploying a KMS layer over QKDN \cite{10.1145/3730575} is cryptographic material allocation, which must resolve the mismatch between the continuous physical key production and the heterogeneous consumption rates of user applications. To prevent buffer starvation during peak demands or entropy waste through overflow during idle periods, the KMS must implement sophisticated pool management strategies. This involves on-demand key reservation mechanisms, strict cryptographic domain isolation to ensure multi-tenant security, and priority-aware allocation frameworks that can provision absolute security via One-Time Pad (OTP) for critical infrastructure while applying symmetric key expansion for lower-priority traffic. Beyond local allocation, the distance limitations of point-to-point QKD links necessitate a multi-hop trusted relay architecture, introducing complex routing and scheduling optimization problems. Unlike traditional IP routing that optimizes for bandwidth or transmission latency, key-aware routing requires the KMS to compute paths based on the real-time availability of accumulated entropy across adjacent quantum links. This demands advanced load-balancing algorithms to prevent key depletion bottlenecks and orchestrate the step-by-step cryptographic relay process efficiently. Furthermore, to mitigate the risk of intermediate node compromise, the KMS must be capable of scheduling multi-path key aggregation schemes, splitting cryptographic material across disjoint routes, which in turn requires highly precise temporal synchronization at the destination endpoint. Finally, this spatial orchestration is underpinned by rigorous key lifecycle management. The KMS must continuously enforce dynamic expiration policies to securely destroy unused cryptographic material, determine optimal rekeying frequencies to guarantee Forward Secrecy, and employ stochastic scheduling models to adaptively respond to the physical fluctuations inherent to the underlying quantum channels.

The system proposed in this document ---hereafter the DKMS (\textit{Distributed Key Management System})--- is a distributed quantum key manager with hybrid links, able to combine PQC and QKD and to offer different levels of security and quality of service.  Its main novelty with respect to earlier proposals lies precisely in this control plane. To materialize this architecture, the specific contributions of our work are twofold:

\begin{itemize}

\item \textbf{Smart Allocation and SDN-inspired Architecture}. The applications that request key material (SAEs) do not draw it down at a smooth, predictable rate; they ask for keys sporadically, in amounts and at instants that are essentially unpredictable from the network's standpoint. Were each request served on demand, QKD's low, distance-dependent generation rate would translate directly into latency and rejected requests. While buffering eases the rate problem locally, it does not solve the global allocation problem across shared links. Rather than relying on hop-by-hop reservation ---where local admission cannot guarantee end-to-end delivery--- this architecture introduces a centralized SDN control plane. The SDN solves a single global instance of \emph{Maximum Concurrent Multi-Commodity Flow} with demand scaling \cite{shahrokhi1990,ahuja1993} over the full matrix of $(s,d)$ pairs. It assigns each flow a rate that simultaneously offsets its consumption and fills its remaining buffer space, driving the synchronized filling of every buffer in the network through a single scalar $\lambda$, while pre-installing the forwarding tables at the intermediate nodes.

\item \textbf{Hybrid Key Distribution}. Both QKD and PQC provide security, but QKD carries constraints that condition its real deployment. The reachable distance over optical fiber is short and the secret-key rate falls quickly with it, so the material a link can actually deliver is scarce and unevenly spread across the network. More fundamentally, without quantum repeaters, multi-hop QKD networks are forced onto trusted relays, reintroducing exactly the trust assumption that QKD set out to remove. A natural way out is to protect the transported key end-to-end with a post-quantum channel, so that no intermediate node ever observes it. In the architecture proposed here, this role falls to the ORR module through its \texttt{max\_hops} parameter: its default setting (\texttt{max\_hops}~$=1$) establishes an end-to-end PQC channel between source and destination, and leaves the physical relay across intermediate nodes entirely to the transport layer.
\end{itemize}

We have validated the proposed architecture under realistic operational conditions on a computing cluster. The experimental environment integrates physical and virtual QKD links into five diverse multi-node network topologies. Under controlled tests, the SDN control plane successfully filled all cryptographic buffers to capacity. Furthermore, under deliberately excessive application demand, the network demonstrated graceful degradation, safely rejecting excess requests through controlled rate-limiting while preserving 
key match integrity, thus proving its resilience and practical viability.

The remainder of this paper is structured as follows. Section \ref{sec:background} reviews current QKD networks, their practical limitations, and positions our contribution within the state of the art. Section \ref{sec:architecture} details the internal architecture of the DKMS and the end-to-end lifecycle of a key. Section \ref{sec:buffers} describes the buffer dynamics and local traffic management, while Section \ref{sec:sdn} formalizes the SDN control plane and the MCMCF-$\lambda$ allocation problem. Section \ref{sec:results} analyzes the experimental results under both buffer-filling and stress-load scenarios. Finally, Section \ref{sec:conclusions} concludes the paper and outlines future research directions.

\section{Current QKD Networks and Their Limitations}
\label{sec:background}

Quantum Key Distribution (QKD) has matured from single-link laboratory experiments into fielded metropolitan and backbone networks. This section reviews representative deployments, articulates the practical limitations that constrain them---above all the rate--distance barrier and the reliance on trusted relays---and surveys the standardization, control, routing, and hybrid-cryptography efforts on which our contribution builds. We then position the proposed DKMS against the specific gap that motivates it: the absence of a centralized, globally fair allocation of scarce key-buffer capacity across all source--destination pairs of a multi-hop network.

\subsection{Deployed QKD Networks}
The first operational QKD network was the DARPA Quantum Network, deployed by BBN Technologies, Harvard, and Boston University across Boston and Cambridge, Massachusetts; it grew to ten nodes and ran continuously from 2004 to 2007, integrating QKD-derived keys into IPsec~\cite{elliott2005darpa}. In Europe, the SECOQC network demonstrated in Vienna in 2008 interconnected several nodes over eight links using heterogeneous QKD technologies and introduced a layered trusted-node architecture~\cite{poppe2008secoqc}. The Tokyo QKD Network, inaugurated on 14 October 2010, integrated six QKD systems into a mesh and demonstrated the world-first secure television conferencing over a 45\,km link (the Otemachi--Koganei span at $14.5$\,dB loss), while placing a dedicated Key Management System (KMS) above the quantum layer~\cite{sasaki2011tokyo}. The SwissQuantum network operated in the Geneva metropolitan area for 21 months, from the end of March 2009 to the beginning of January 2011, to validate long-term robustness in the field~\cite{stucki2011swissquantum}. The Cambridge/UK Quantum Network demonstrated a high-rate metropolitan ring coexisting with 100\,Gb/s data traffic on installed fiber, together with a robust key-delivery layer~\cite{dynes2019}.

The most extensive deployment is in China, where the Beijing--Shanghai backbone, completed in 2017, spans $2\,032$\,km through 32 trusted-relay nodes; it was integrated with the Micius satellite's two satellite-to-ground free-space links and more than 700 fiber links to form an integrated space-to-ground network reaching $4\,600$\,km and serving more than 150 industrial users, including state and local banks, municipal power grids, and e-government platforms~\cite{chen2021}. In Europe, the Madrid Quantum Network (MadQCI) has deployed 28 QKD modules from five manufacturers across nine production nodes of a multi-tenant telecommunications network under a software-defined control plane~\cite{martin2024madqci}, building on the software-defined QKD engineering pioneered by the same groups~\cite{aguado2019sdn}. These field efforts are coordinated at continental scale by the OpenQKD testbed program and the EuroQCI initiative---a pan-European quantum communication infrastructure whose 2019 Declaration has been signed by all 27 EU Member States~\cite{openqkd,euroqci}. Comprehensive surveys document this evolution~\cite{cao2022qinternet,mehic2020networking}.

\subsection{The Rate--Distance Barrier and Trusted Relays}
The defining limitation of fiber QKD is that a quantum signal cannot be amplified, so the secret-key rate decays with channel loss. In the low-loss regime the achievable rate is well captured by a simplified attenuation model,
\begin{equation}
R(L) = R_0 \, 10^{-\alpha L / 10},
\label{eq:rate}
\end{equation}
where $\alpha$ is the fiber attenuation (typically $0.2$\,dB/km) and $R_0$ the back-to-back rate. This exponential decay is not merely an engineering artifact: Pirandola \emph{et al.} established the repeaterless secret-key capacity bound (the PLOB bound), a fundamental upper limit of $-\log_2(1-\eta)$ secret bits per channel use for a lossy channel of transmittance $\eta$, which no point-to-point protocol can exceed without quantum repeaters~\cite{pirandola2017plob}. Consequently, key rates are low and strongly heterogeneous across a network: high-rate systems reach a sustained $13.72$\,Mb/s over a $2$\,dB channel~\cite{yuan2018mbps}, whereas the twin-field protocol, which improves the scaling from $\eta$ to $\sqrt{\eta}$, has extended distance to a record $1\,002$\,km but at a vanishing rate of only $0.0034$\,bit/s~\cite{liu2023tfqkd1000}. Rate and reach are therefore in fundamental tension.

Because practical quantum repeaters remain confined to the laboratory---no group has yet demonstrated a repeater outperforming direct transmission under realistic conditions, and they are still regarded as a subject of fundamental research~\cite{quantumrepeater}---operational networks overcome the distance limit exclusively through \emph{trusted relays}. Keys are relayed hop-by-hop, typically by one-time-pad (XOR) combination of per-link keys at each intermediate node~\cite{tysowski2018}. This has a critical security consequence: every intermediate node sees the relayed key in the clear and must be fully trusted, and each relay consumes scarce link key material, which can be depleted and cause key-exchange failure. Managing the mismatch between key generation and consumption---through per-link buffering, storage, and routing of key material---is thus a central challenge of QKD networking.

\subsection{Standardization and SDN-Based Control}
A standardization landscape has emerged to make these networks interoperable and manageable. ETSI GS QKD 014 defines a REST-based key-delivery API over HTTPS by which Secure Application Entities (SAEs) retrieve keys from Key Management Entities~\cite{etsi014} (hereafter ETSI~014), while ETSI GS QKD 015 specifies the control interface between software-defined QKD nodes and an SDN controller, and ETSI GS QKD 018 the orchestration interface~\cite{etsi015,etsi018}; the draft ETSI GS QKD 020 adds a horizontal, REST-based Interoperable KMS API by which key managers pass keys to one another~\cite{etsi020}. In parallel, the ITU-T Y.3800 series provides the reference framework, functional architecture, and key-management model for QKD networks (QKDNs)~\cite{itu_y3800}. Software-defined networking has become the dominant paradigm for integrating QKD into carrier infrastructure: the Madrid work demonstrated an SDN control plane that abstracts heterogeneous QKD devices and exposes a network-wide view for key management~\cite{aguado2019sdn,martin2024madqci}. These architectures centralize \emph{control and management} but, as discussed below, they do not centralize the \emph{optimization of key-generation and buffer capacity} across the whole network.

\subsection{Key Routing, Resource Allocation, and Hybrid QKD--PQC}
Given scarce and depletable key material, routing and resource allocation are active research areas. Secret-key-aware routing selects relay paths using link cost functions based on residual key, generation rate, and consumption rate~\cite{yang2018secretkey}; resource-aware and collaborative routing recompute paths from current key availability to raise the key-relay success rate~\cite{yu2022skpcr}; and further schemes address routing over networks combining trusted nodes and repeaters~\cite{amer2020routing}. Information-centric and joint routing/channel/key-rate formulations optimize the acceptance ratio of individual key-rate requests~\cite{zhang2023icqkd}. These approaches are predominantly \emph{per-demand} or \emph{per-request}: they route or admit one commodity at a time, rather than jointly and fairly allocating capacity across all pairs. The classical operations-research tools for the latter---multi-commodity flow and the maximum concurrent flow problem, which maximizes a common throughput fraction subject to capacity constraints---are well established~\cite{shahrokhi1990,ahuja1993} and efficiently solvable with modern LP solvers~\cite{huangfu2018highs}, but their systematic application to QKD key buffers is only nascent. The concurrent-flow idea has been proposed for a QKD key-management architecture~\cite{tysowski2018} and for max-min allocation in a satellite (LEO/GEO) trusted-repeater setting~\cite{grillo2021}, yet a general, fairness-centric global allocation of key-buffer capacity across all source--destination pairs of a terrestrial multi-hop network remains largely unexplored.

Finally, because QKD alone is range-limited and depends on authenticated classical channels, hybrid QKD+PQC designs that combine information-theoretic and computational security are increasingly favored: concatenating a QKD key with a post-quantum (e.g., ML-KEM) shared secret yields a composite key that remains secure as long as either component holds~\cite{zeng2024hybrid,braun2023hybrid}.

\subsection{Positioning of This Work}
The literature reveals a consistent pattern: deployed QKD networks rely on hop-by-hop trusted relay with local, per-link key management; SDN control planes centralize configuration and monitoring but not capacity optimization; and routing/allocation research optimizes individual demands rather than the network as a whole. The scarce, distance-dependent, and heterogeneous key-generation capacity of a multi-hop QKD network is therefore rarely managed by a single global optimizer that fairly shares it among all communicating pairs. Our DKMS addresses precisely this gap. It couples per-peer key buffers and standardized ETSI~014 application interfaces with an SDN control plane that solves a Maximum Concurrent Multi-Commodity Flow with demand scaling (MCMCF-$\lambda$) linear program, computing a globally fair allocation of key-buffer capacity across every source--destination pair. Hop-by-hop OTP key relay through trusted nodes is retained for reach, while onion routing with ML-KEM provides end-to-end protection of relayed keys, and a hybrid QKD+PQC scheme hedges the residual trust and range assumptions. In contrast to prior per-demand routing and to control planes that manage without globally optimizing, our approach treats network-wide key provisioning as a single concurrent-flow optimization, turning the physical scarcity captured by~\eqref{eq:rate} and the PLOB bound into an explicit, solvable allocation problem.

\section{DKMS Architecture}
\label{sec:architecture}

The system proposed in this work ---the DKMS--- is not a single component but a distributed key manager assembled from four cooperating modules: \textbf{KMS}, \textbf{ORR}, \textbf{QKC} and \textbf{SDN}. The first three make up the \emph{data plane}, which generates, protects and transports the keys; the SDN forms the \emph{control plane}, which decides at what rate keys are requested and along which routes they travel. This internal division is invisible from the outside: users or applications, termed SAEs (\emph{Secure Application Entities}), simply connect to the DKMS through a standard ETSI~014 interface and obtain symmetric keys secure against quantum adversaries. Each physical node is a hybrid node, able to derive key material both from a QKD link and from PQC. A remark on naming: \emph{DKMS} designates both the system as a whole and, where the context makes it unambiguous, a single node of it (the KMS--ORR--QKC triple). The module inside a node that owns the key buffers and the ETSI~014 interface is its KMS; when later on we say that ``the DKMS'' reports state or polls a rate, it is this module acting on the node's behalf.

\begin{figure*}[t]
    \centering
    \includegraphics[width=0.80\textwidth]{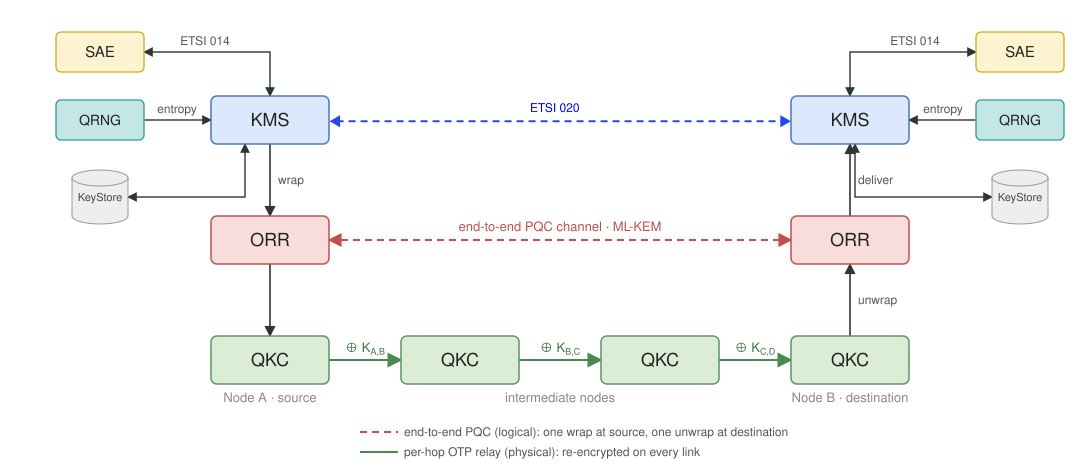}
    \caption{Architecture of a DKMS node and the end-to-end journey of a \emph{transport} key into the per-peer buffers. The ORR wraps it end to end for the destination ORR (dashed logical channel, via ML-KEM), while the QKCs relay it hop by hop with a per-link OTP (physical path), so intermediate nodes never hold it. Session keys handed to SAEs do not traverse this relay: they travel over the direct KMS-to-KMS ETSI~020 channel, protected by one-time pad with a buffered transport key.}
    \label{fig:e2e_flow}
\end{figure*}

Rather than describing each module in isolation, it is easier to follow the life of a key across the system; Figure~\ref{fig:e2e_flow} traces that journey. The key observation is that the system handles two kinds of keys, produced and consumed on decoupled timescales: \emph{transport} keys, relayed across the network ahead of demand and accumulated in per-peer buffers, and \emph{session} keys, generated on demand and handed to the SAEs. Continuously and ahead of any request, each \textbf{KMS} generates transport-key material ---drawing fresh randomness from a quantum random number generator (emulated QRNG)--- and pushes it toward every peer. To do so it gives the material to the \textbf{ORR}, which ---in the default mode--- wraps it in an end-to-end post-quantum channel addressed to the destination ORR, so that no intermediate node can read it. The wrapped material then enters the \textbf{QKC}, the module that actually moves it across the network: each QKC along the path performs a single One-Time Pad (OTP) operation per link, decrypting the incoming key payload with the material it shares with the previous neighbor, choosing the next hop, and re-encrypting with the material of the next link. At the destination node its QKC delivers the packet to the local ORR, which peels the post-quantum layer and recovers the material, and the destination KMS deposits it in the buffer it keeps for that peer. This filling process runs on its own, so that at any moment each pair of KMSs already holds a reserve of shared transport keys on both ends.

Only then, and independently, does consumption happen: when a local SAE asks its KMS over the ETSI~014 interface for a key shared with a remote peer, the KMS generates the \emph{session} key on the spot, protects it by one-time pad with a fresh transport key drawn from the corresponding buffer ---consuming exactly one buffered key per delivery--- and hands it to the destination KMS over a direct, mutually authenticated (mTLS) HTTP/2 connection, the ETSI~020 interface described below. This is precisely what buffering buys: what the reserve removes from the critical path is the slow, multi-hop QKD relay, so the request is served at the cost of a single direct network round-trip between the two KMSs while inheriting the protection under which that transport material was relayed ---per-link OTP with QKD link keys against outside observers, and the ORR's end-to-end ML-KEM wrap against the relays themselves. Overseeing the whole process, but never touching the keys themselves, the \textbf{SDN} control plane decides beforehand at what rate each KMS replenishes each transport-key buffer and which route the QKCs follow.

With the path in mind, each module can now be described in turn. The KMS and the SDN are only outlined here, since their internals are the subject of Sections~\ref{sec:buffers} and~\ref{sec:sdn} respectively; the ORR and the QKC, which have no dedicated section, are described in full.

\paragraph{\textbf{KMS}} \emph{The key creation and distribution layer.} It produces key material from a quantum random number generator (QRNG), persists it in a local KeyStore, and delivers keys to the local SAEs through an ETSI~014 interface~\cite{etsi014}. This ETSI~014 \texttt{Get key} endpoint is the only interface the SAEs ever see; the movement of key material \emph{between} KMSs instead follows a separate, horizontal interface, the REST-based Interoperable KMS API of the draft ETSI GS QKD~020 specification~\cite{etsi020}, whose purpose is precisely to pass keys from one key manager to another. When a local SAE requests a key shared with a peer SAE hosted on another node, the source KMS generates the \emph{session} key from its QRNG, protects it by one-time pad with a transport key taken from the per-peer buffer, and delivers it to the destination KMS over this ETSI~020 interface ---a direct KMS-to-KMS HTTP/2 exchange under mutual TLS; the ORR/QKC relay of Figure~\ref{fig:e2e_flow} is not on this path, its role being to keep the transport-key buffers replenished in the background--- so that the destination KMS can file the key under its \texttt{key\_ID} and its SAE recover it with an ETSI~014 \texttt{Get key with key IDs} request. Because the ETSI~014 \texttt{Get key} request carries an optional \texttt{additional\_slave\_SAE\_IDs} list, this same mechanism extends naturally to multicast key sharing: a single SAE can obtain one key shared simultaneously with several SAEs that reside on distinct KMSs, and the DKMS fulfills it by replicating the ETSI~020 delivery ---each copy wrapped with a transport key of the corresponding peer--- to every KMS involved. To decouple these deliveries from the slow, bursty generation of QKD, each KMS keeps a continuously fed buffer of transport keys for every peer; how those buffers are dimensioned, filled and policed is the subject of Section~\ref{sec:buffers}.

\paragraph{\textbf{ORR} (Onion Routing Router)} \emph{The security layer.} It abstracts the trust policy over the intermediate nodes: whenever a KMS pushes transport-key material to another node, it always hands it first to the ORR, whose job is to protect that material end to end with post-quantum cryptography so that no intermediate relay can read ---or tamper with--- it. This behavior is governed by an optional \texttt{max\_hops} parameter. With \texttt{max\_hops}~$= 0$ no cryptographic wrap is applied and the key is passed straight to the local QKC, so security rests entirely on trusting the intermediate nodes and on the hop-by-hop OTP \emph{key relay} of the QKC. With \texttt{max\_hops}~$= 1$ ---the default used to populate the buffers--- the ORR establishes an end-to-end PQC channel between the source and destination nodes, so that no intermediate node ever observes the transported key. With \texttt{max\_hops}~$= -1$ a full hop-by-hop \emph{onion} wrap is applied along the whole path, adding one protection layer per intermediate node; any integer greater than~$1$ acts as a lighter variant that caps the number of layers. The onion layers are not a fresh post-quantum encapsulation per message: each pair of ORRs establishes a long-lived shared secret once, by means of a post-quantum key-encapsulation mechanism (ML-KEM), and then derives a lightweight per-key symmetric layer from it. The post-quantum primitive thus protects the establishment of the inter-ORR secret, while per-key wrapping stays cheap. Having applied (or not) the cryptographic wrapping, the ORR hands the packet to the QKC for transport across the network.

\paragraph{\textbf{QKC} (Quantum Key Channel)} \emph{The transport layer.} It interconnects the network's links ---whether backed by a real QKD channel or by a PQC-based substitute that plays the same role--- and behaves as a router, or switch, of keys: given a key destined for a particular node, it decides to which physical neighbor to forward it. To do so it consults a \textbf{local forwarding table indexed by destination} that the SDN has pre-installed for it (see Section~\ref{sec:sdn}). When the table holds a single next hop the choice is deterministic; when the destination is served by several paths, the QKC picks the next hop for each key by hashing its identifier against the SDN's WCMP weights, so the traffic split matches the weights statistically while giving each key a deterministic path. Before sending, the QKC encrypts the key by \emph{One-Time Pad} (OTP) with the material of the chosen link: writing $K_{a,b}$ for the material shared between neighboring QKCs $a$ and $b$, the key payload leaving over that link is $m = S \oplus K_{a,b}$. The receiving QKC recovers $S = m \oplus K_{a,b}$ and, if it is not the final node, re-encrypts with the material of the next link.

\paragraph{\textbf{SDN}} \emph{The orchestration layer.} With a global view of the topology, the buffer levels and the observed demand, it decides at what rate each KMS must request keys from every peer and pre-computes the routes that the QKCs apply. It is the only module with a network-wide picture. 

Two properties of this design are worth making explicit. First, the ORR and the QKC operate at different granularities: the ORR protection is applied end to end ---one wrap at the source, one unwrap at the destination--- whereas the OTP relay is per hop, re-applied on every link the key crosses. This is exactly what lets the intermediate nodes move the transport-key material without ever seeing it: they relay the PQC-wrapped payload, and none of them holds the secret needed to open it. Second, in this default end-to-end mode neither the KMS nor the ORR needs to know the topology, and no module queries the control plane while a key is in flight: each QKC forwards solely from the table the SDN pushed beforehand, so the entire data-plane journey completes without a single round-trip to the SDN. Only in the full or truncated \emph{onion} modes does the ORR additionally receive the ordered sequence of relay nodes; the per-link forwarding always remains the QKC's responsibility.

\section{Buffers and Traffic Management}
\label{sec:buffers}

One of the main problems of current QKD networks is that the key generation rate is low. One way to partially mitigate this limitation is to generate and store key material continuously ---even when there is no immediate consumption--- in per-peer buffers, so that requests are protected with material already on hand rather than waiting on the instantaneous link rate. That said, in a network of $n$ DKMSs a mechanism is additionally needed to \emph{globally} distribute capacity among all pairs, so that no buffer fills up or empties out in an unbalanced way.

Each DKMS therefore has an explicit buffer $B_{s,d}$ with each of the other DKMSs in the network, with a reference capacity $B_{\max}$ that the control plane treats as the fill target. These buffers hold the \emph{transport} keys: each key delivered to an SAE consumes exactly one of them to wrap the session key, so SAE demand drains the buffers one for one. Management is organized at two levels. At the \textbf{local level}, inside each DKMS, a \emph{token bucket} per SAE polices the request stream and shares the per-peer rate fairly among the SAEs that are simultaneously active, preventing any one of them from starving the others. At the \textbf{global level}, the SDN controls \emph{at what rate} each DKMS requests keys from each peer and \emph{by which path} they are transported, by solving an MCMCF-$\lambda$ LP over the entire matrix of pairs with a single synchronization scalar.

\subsection{Buffer dynamics and associated commodity}
\label{subsec:buffer_dynamics}

Each buffer $B_{s,d}$ is modeled as a first-order linear system whose level $L_{s,d}(t)$ evolves according to
\[
    \frac{dL_{s,d}}{dt} = r_{s,d}(t) - \delta_{s,d}(t),
\]
where $r_{s,d}(t)$ is the key injection rate that the network delivers to destination $d$ from source $s$ (controlled by the SDN, Section~\ref{sec:sdn}) and $\delta_{s,d}(t)$ is the aggregate consumption rate by the local SAEs of DKMS $s$ that are already authorized to talk with peers served by $d$. The identification of each ordered pair $(s,d)$ with a \emph{commodity} $k$ of the MCMCF-$\lambda$ problem is direct: $s_k = s$, $t_k = d$, $B_k = B_{\max}$, $L_k = L_{s,d}(t)$, and $\delta_k = \hat{\delta}_{s,d}(t)$ with the latest EWMA estimate.

The \textbf{remaining space} of the buffer is additionally defined as
\[
    R_{s,d}(t) \;=\; B_{\max} - L_{s,d}(t) \;\geq\; 0.
\]
This quantity is the only one the model needs to articulate priority: the emptier the buffer (the larger $R_{s,d}$), the greater the filling component that the SDN assigns to the corresponding flow. No discrete classes or explicit thresholds are required: priority emerges as a continuous linear dependence of the LP, described in the next Section.

\paragraph{Demand estimation $\hat{\delta}_{s,d}$}
The DKMS internally measures the number of keys drained by SAEs in each one-second window and maintains an exponential moving average with factor $\alpha \approx 0.2$:
\[
    \hat{\delta}_{s,d}(n+1) \;=\; (1-\alpha)\,\hat{\delta}_{s,d}(n) \;+\; \alpha \cdot \text{drain}(n),
\]
where $\text{drain}(n)$ counts the key material requested by the SAEs of commodity $(s,d)$ in window $n$, accumulated \emph{before} the admission \emph{token bucket}. Counting attempted rather than served demand makes $\hat{\delta}_{s,d}$ reflect the load that the SAEs try to exert ---including requests that are momentarily rate-limited--- and not only the load the network currently serves. This $\hat{\delta}_{s,d}$ is what the DKMS reports to the SDN.

\subsection{Local per-SAE control}

To prevent an aggressive SAE from monopolizing its DKMS's buffer, each SAE associated with $B_{s,d}$ is policed by an independent \emph{token bucket} parameterized from the rate $r_{s,d}$ dictated by the SDN. The refill rate is the per-SAE share $X_k = r_{s,d}(t) / N_{\text{SAEs}}$, where $N_{\text{SAEs}}$ is the number of SAEs that have recently been active against that peer, so that admitting a new SAE proportionally reduces the share of the others. The bucket size is bounded by both a burst allowance and the \emph{current buffer level} $L_{s,d}(t)$: when the offered load exceeds the locally available key material, the request is rejected as back-pressure (HTTP~$429$) rather than served from an empty buffer, and the cost charged to the bucket grows with the amount of key material requested. There is no separate scheduler: fairness among the SAEs of the same DKMS emerges from these per-SAE buckets, recomputed against the live count of active SAEs on every admission, without interfering with the global allocation among DKMSs discussed in the following section.

\section{SDN Control Plane Based on MCMCF-$\lambda$}
\label{sec:sdn}

The proposed architecture moves \emph{all} of the decision-making to the control plane: the SDN maintains the graph and a complete matrix of commodities $(s,d)$, solves the capacity allocation globally, \textbf{pre-installs the WCMP forwarding tables} in the intermediate QKCs, and \textbf{dictates to each DKMS the rate} at which it must request keys. Upon changes in the levels reported by the DKMSs or in the estimated capacity of the links, the SDN recomputes and the DKMSs read the new allocation on their next periodic \emph{poll}. Figure~\ref{fig:sdn_control_loop} summarizes this closed control loop: buffer state flows up to the global LP, and the resulting routes and rates are pushed back down to the data plane.

\begin{figure*}[t]
    \centering
    \includegraphics[width=0.85\textwidth]{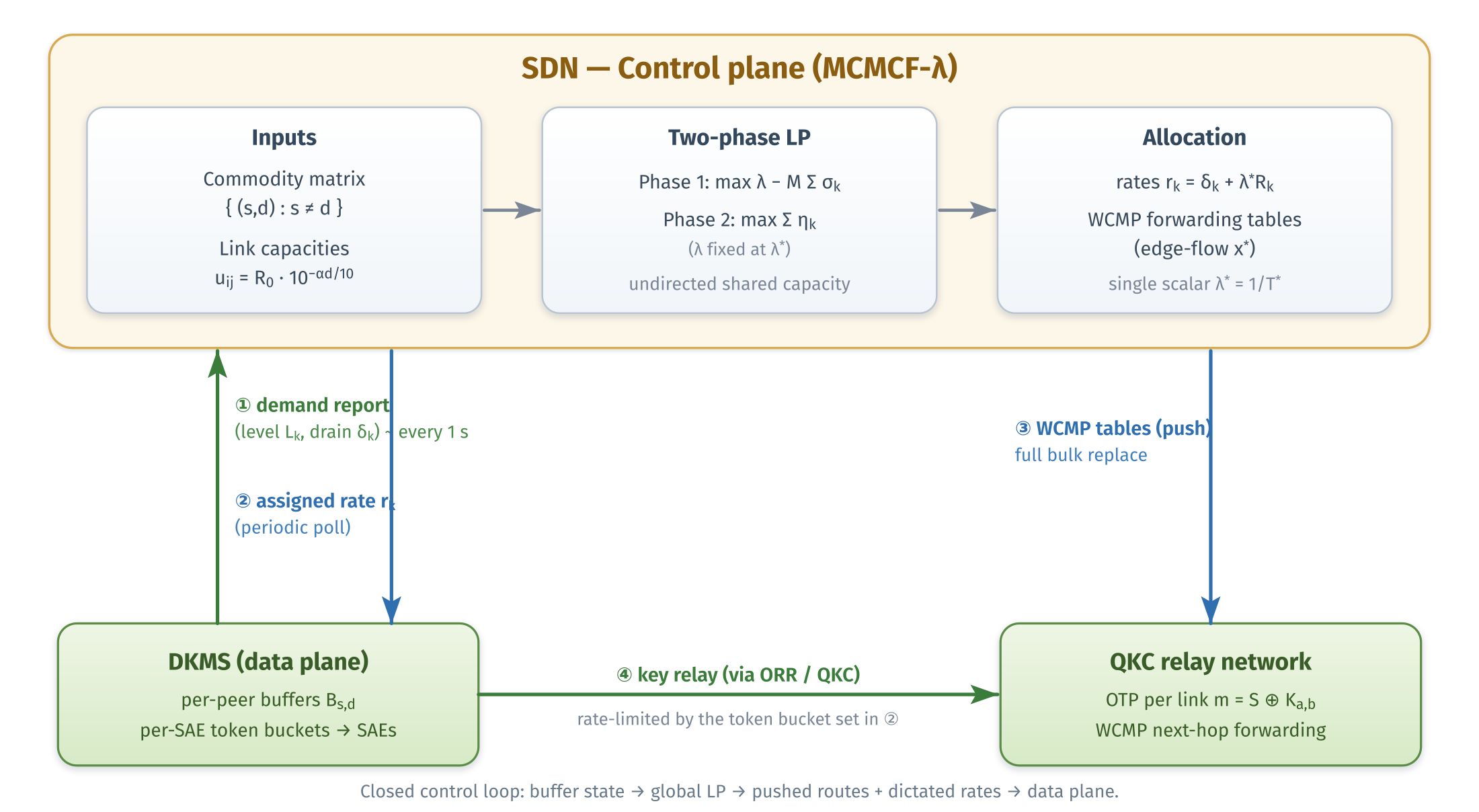}
    \caption{The MCMCF-$\lambda$ control loop. Each DKMS reports its per-commodity buffer state (level $L_k$, drain $\delta_k$); the SDN assembles the commodity matrix and the attenuated link capacities $u_{ij}$, solves the two-phase LP, and pushes the WCMP forwarding tables to the QKCs while publishing the per-commodity rates $r_k = \delta_k + \lambda^\star R_k$ (base form; refined in Section~\ref{subsec:lex}) that each DKMS polls. The single scalar $\lambda^\star$ characterizes the whole network's saturation state.}
    \label{fig:sdn_control_loop}
\end{figure*}

\subsection{Formulation as MCMCF with demand scaling}
\label{subsec:formulacion}

The problem is modeled as an instance of \emph{Maximum Concurrent Multi-Commodity Flow} \cite{shahrokhi1990,ahuja1993} over an undirected graph in \emph{edge-flow} formulation, with a single scalar $\lambda$ that controls the filling speed of all the network's buffers in a synchronized manner.

\paragraph{Graph} $G = (V, E)$ \emph{undirected}, where $V$ is the set of physical QKCs and $E$ the set of links, with capacity $u_{ij}$ (in keys/s) per link $\{i,j\} \in E$. The graph is undirected because each QKD link is a single physical quantum channel that generates \emph{shared} material in a single event: both ends receive the same bytes and the capacity $u_{ij}$ is distributed among the flows crossing $\{i,j\}$ in either direction. To facilitate the formulation, each link is split into two directed arcs $A = \{(i,j),(j,i) : \{i,j\} \in E\}$, coupled by the shared-capacity constraint.

\paragraph{Commodities} $K = \{\,k = (s_k,t_k) : s_k,t_k \in V_{\text{DKMS}}, s_k \neq t_k\,\}$, indexed by ordered pair of DKMSs. Each commodity $k$ is associated with the buffer $B_{s_k, t_k}$ and is characterized by four scalars: a buffer capacity $B_k > 0$ (constant within each solve), the current level $L_k \in [0, B_k]$ reported by DKMS $s_k$, the EWMA drain rate $\delta_k \geq 0$ also reported by DKMS $s_k$, and the remaining buffer space $R_k = B_k - L_k$. The inverse pair $(t_k, s_k)$ is a distinct commodity because the buffers $B_{s,d}$ and $B_{d,s}$ are independent, even though both share the physical capacity $u_{ij}$ of the links they traverse.

\paragraph{Variables} For each commodity $k$ and each directed arc $(i,j) \in A$, a variable $x^{k}_{ij} \geq 0$ representing the amount of flow of $k$ over the arc. Additionally, a scalar $\lambda \geq 0$ that acts as the inverse of the synchronized filling time.

\paragraph{Rate--scale relation} If all buffers are to fill exactly at time $T = 1/\lambda$, the total rate delivered to commodity $k$ must be
\begin{equation}
    r_k \;=\; \delta_k \;+\; \lambda\, R_k.
    \label{eq:r_k}
\end{equation}
The interpretation is direct: $\delta_k$ compensates the consumer's drain and $\lambda R_k$ provides the filling flow, proportional to the remaining space. Buffers with more space receive more filling rate; already-full buffers ($R_k = 0$) receive only $\delta_k$.

\paragraph{Constraints} Flow conservation per commodity and node, with demand $\pm r_k$ at source/destination:
\begin{equation}
    \sum_{j:\,(i,j)\in A} x^{k}_{ij} \;-\; \sum_{j:\,(j,i)\in A} x^{k}_{ji} \;=\; b^{k}_i, \;\; \forall k\in K,\,i\in V,
    \label{eq:flow}
\end{equation}
with $b^{k}_i = +r_k$ if $i = s_k$, $b^{k}_i = -r_k$ if $i = t_k$, and $b^{k}_i = 0$ otherwise, and \textbf{undirected shared capacity} per link:
\begin{equation}
    \sum_{k\in K} \bigl(x^{k}_{ij} + x^{k}_{ji}\bigr) \;\leq\; u_{ij}, \quad \forall\, \{i,j\}\in E.
    \label{eq:cap}
\end{equation}
Note that \eqref{eq:cap} couples \emph{both directions} of the same physical link over all commodities. This constraint is what distinguishes the model from the classical directed scheme and is indispensable for faithfully representing the shared nature of the QKD channel.

\paragraph{Objective and synchronization property} Maximize $\lambda$:
\begin{equation}
    \boxed{\;\;\max\;\lambda \quad \text{s.t.\ } \eqref{eq:flow},\;\eqref{eq:cap},\; x^{k}_{ij},\lambda \geq 0\;\;}
    \label{eq:lp}
\end{equation}
Substituting \eqref{eq:r_k} into \eqref{eq:flow}, the commodity sees $\lambda R_k$ as a demand term at its source, and the LP is solved in polynomial time. The current prototype uses \texttt{microlp}, an in-memory revised-simplex solver; for larger instances a sparse interior-point backend such as HiGHS \cite{huangfu2018highs} can be substituted without changing the model.

\begin{theorem}[Filling synchronization]
\label{prop:sync}
Let $(\lambda^\star, \{x^{k\star}_{ij}\})$ be an optimal solution of \eqref{eq:lp} and $r_k^\star = \delta_k + \lambda^\star R_k$ the induced rates. Then all buffers reach their capacity $B_k$ simultaneously at the instant $T^\star = 1/\lambda^\star$.
\end{theorem}
\begin{proof}
The buffer level evolves as $L_k(t) = L_k + (r_k^\star - \delta_k)\,t = L_k + \lambda^\star R_k \cdot t$; evaluating at $t = T^\star$ yields $L_k(T^\star) = L_k + R_k = B_k$ for all $k$. The equality of times $T_k = T^\star$ is obtained algebraically without the need to impose classes or explicit priorities.
\end{proof}

\subsection{Lexicographic refinement of residual bandwidth}
\label{subsec:lex}

The solution of \eqref{eq:lp} saturates the set of links that constitute the system's bottleneck, but leaves free capacity in the others. That slack can be exploited to over-fill buffers with alternative paths through non-critical arcs, reducing their individual filling time without affecting the joint time $T^\star$. A variable $\eta_k \geq 0$ is introduced per commodity with upper bound $R_k / T_{\text{replan}}$ that prevents filling more than fits in a re-computation horizon. The effective rate becomes $r_k = \delta_k + \lambda R_k + \eta_k$.

Instead of the classical method $\max \lambda + \varepsilon \sum_k \eta_k$, which fails numerically when $\lambda$ is small relative to the range of $\eta$, the solver solves two chained LPs with exact lexicographic ordering. \textbf{Phase~1} maximizes $\lambda - M\sum_k \sigma_k$, where $\sigma_k \in [0, \delta_k]$ are slack variables that absorb the unserviceable drain demand when it exceeds the minimum cut, and $M = 10^{3}$ is a strong penalty that forces the solver to nullify the deficit before maximizing $\lambda$. \textbf{Phase~2} maximizes $\sum_k \eta_k$ with $\lambda$ \emph{substituted} literally by its phase-1 optimum $\lambda^{\star}$, rather than imposed as a lower-bound constraint $\lambda \geq \lambda^{\star} - \epsilon_{\text{num}}$. The substitution removes $\lambda$ from the variable set of phase~2 entirely and is intended to suppress a numerical pathology observed with simplex-based solvers, which intermittently reported phase~2 as infeasible when $\lambda^{\star}$ was small relative to the dynamic range of $\eta_k$. The constraint structure of phase~2 inherits from phase~1 by replacing every occurrence of $\lambda$ with the scalar $\lambda^{\star}$. The final rate that the DKMS receives is $r_k = (\delta_k - \sigma_k)_+ + \lambda^\star R_k + \eta_k^\star$.

\paragraph{Graceful degradation under overload} The variables $\sigma_k$ resolve the pathological case in which the aggregate demand $\sum_k \delta_k$ exceeds the capacity of the minimum cut between sources and destinations. Without slack, the LP is infeasible and the SDN is forced to return $r_k = 0$ for all commodities (\emph{cascade fail}). With slack, the LP is always feasible: the network serves what it can, records the deficit as $\sigma_k$, and the DKMSs adjust their \emph{token buckets} to the deliverable pace; the SAEs see \emph{back-pressure} (HTTP 429) instead of a control-plane failure.

\paragraph{Special case $R_k = 0$ for all $k$} If all buffers are full, $\lambda$ has a null coefficient in all constraints and $\max \lambda$ is unbounded. The solver cuts off this case by applying $\lambda = 0,\, r_k = \delta_k$ (pure maintenance), avoiding the pathology without perturbing the rest of the control flow.

\paragraph{Behavior at the evaluated scale} The literal substitution removes the documented source of drift, but with the simplex backend used in the prototype, phase~2 is still reported infeasible in roughly $70\,\%$ of the solves at the commodity counts of the evaluated topologies ($|K| = 380$). The failure matters operationally, because phase~2 produces both outputs consumed downstream: the $\eta_k$ bonuses and the edge-flow that the SDN projects onto the WCMP tables ---the phase-1 flows are discarded by the solver, which retains only $\lambda^\star$ and the slacks $\sigma_k$. On a successful solve the network therefore runs with multipath WCMP forwarding and $\eta$-boosted rates; on a failed one it falls back to the phase-1 rate $(\delta_k - \sigma_k)_+ + \lambda^\star R_k$ and to plain shortest-path forwarding until the next successful solve. The evaluation campaign of Section~\ref{sec:results} ran under this intermittent regime ---roughly one solve in three succeeding--- and the bonus is far from cosmetic: in earlier trials with phase~2 disabled, buffer saturation was roughly an order of magnitude slower. A more robust interior-point backend is expected to make phase~2 reliable at scale.

\subsection{Forwarding tables in the QKC (WCMP)}
\label{subsec:forwarding}

The \emph{edge-flow} solution of the LP is projected onto the forwarding tables of each QKC. For each transit QKC and each destination, the flows of all commodities that share that destination are aggregated per outgoing neighbor, yielding a set of \emph{Weighted-Cost Multi-Path} (WCMP) next hops with their relative weights. The result is, for each QKC, a compact table indexed by destination that maps every destination to a weighted list of next-hop neighbors. An entry has a single next hop when the LP concentrates all the flow toward that destination on one path, and several when the solution is genuinely \emph{multipath} over the graph; in the latter case each key is assigned its next hop by a hash of its identifier over the weights, so the split matches the WCMP weights statistically. Installation is performed by \emph{push} from the SDN after each recomputation: the SDN replaces the full table of every affected QKC in a single bulk update. When phase~2 of the solver fails, no edge-flow is available and the SDN instead installs plain shortest-path tables until the next successful solve. If a QKC restarts and loses its table, the next recomputation reinstalls it. A finer mechanism, in which each entry carries a TTL and the QKC performs \emph{pull-on-miss} against the SDN, is left for future work; the current prototype relies solely on the push.

\subsection{Rate polling by the DKMS}

The rate at which each \emph{ingress} DKMS requests keys from each peer is obtained by periodically \emph{polling} the SDN. The response enumerates all peers for which the DKMS is the source and, for each, splits the assigned capacity into its two operational directions (the encryption side and the decryption side), reflecting the independence of the two buffers $B_{s,d}$ and $B_{d,s}$ that a pair of DKMSs maintains. It also carries a topology-version tag that lets the DKMS detect when the SDN has republished a new topology snapshot and invalidate any local state derived from the previous one. Locally the DKMS enforces the encryption-side value as the token-bucket ceiling for its own emission toward each peer; the decryption-side value is the one the peer enforces symmetrically for the reverse buffer. The choice of \emph{polling} over \emph{callback} is deliberate: the period (of the order of one second, configurable) is short enough to react to capacity redistributions, and it simplifies the operation of the control plane by avoiding return credentials, retries, and potentially unreachable \emph{back-channels}.

\subsection{Demand reporting by the DKMS}
\label{subsec:demand}

Each DKMS periodically publishes (about once per second) the state of all the buffers for which it is the source: for every commodity it reports the current level $L_k$, the reference capacity $B_k$, the EWMA drain rate $\delta_k$ and a timestamp. The SDN keeps the latest report per pair $(s,d)$ and uses it as the input $L_k$, $B_k$, $\delta_k$ of the LP in the next recomputation.

Unlike a scheme based on the \emph{push} of discrete events, the report is always complete: each message overwrites the entire buffer state of the emitting DKMS. This makes the control plane \emph{convergent by construction} in the face of any transient failure: it suffices for the next report to reach the SDN for the state to reconcile, without the need for explicit \emph{retry} mechanisms or \emph{idempotency keys}.

\subsection{Control loop and response to change}
\label{subsec:controlloop}

The SDN controller runs a periodic loop that, in each iteration, reads the latest demand reports and the capacity matrix $u_{ij}$ of the current topology snapshot, builds the LP inputs (commodities, arcs and capacities), solves the two LP phases described in Section~\ref{subsec:lex} on a dedicated \emph{worker} thread so as not to block the request handlers, projects the resulting edge-flow onto the per-QKC WCMP representation, publishes the new rate snapshot and pushes the new WCMP tables in parallel to the affected QKCs. When a commodity enters overload, its slack $\sigma_k$ absorbs the deficit and its delivered rate drops proportionally, leaving the free capacity available for the rest via $\lambda$. When a buffer fills ($R_k \to 0$), its contribution to $\lambda$ disappears and the rest of the network sees the released slack. There is no \emph{tear-down} or re-establishment: only recompute and publish.

\subsection{Bootstrap of the commodity matrix}
\label{subsec:bootstrap}

Since the SDN starts with no demand information, if it were to wait for the first report to solve the LP, the DKMSs that polled for their rate before then would see null rates. To avoid this blockage, the solver synthesizes a commodity for \emph{every} ordered pair of distinct DKMSs, with defaults $(L_k = 0,\, B_k = 4\,096,\, \delta_k = 0)$ ---$4\,096$ keys being the DKMS factory default, overwritten by the first real report---; pairs without connectivity simply receive $r_k = 0$ from the LP. As the real reports arrive, those defaults are overwritten entry by entry. Thus every commodity that \emph{awakens} already has capacity assigned and begins to flow without renegotiating.

\subsection{Recomputation triggers}

The control loop of Section~\ref{subsec:controlloop} runs periodically with a configurable period (of the order of one second), and the resulting ticks are coalesced through a \emph{debounce} window whose floor the controller auto-tunes from the size of the topology. Since the empirical cost of one solve grows roughly as $\mathcal{O}(|K|^2 \cdot |E|)$, a fixed window short enough for small topologies would stack recomputes on large ones faster than they can finish; the controller therefore raises the floor of the debounce as the network grows so that consecutive solves do not overlap. Demand reports update the registry continuously and are absorbed by the next scheduled solve. Topology mutations (addition or removal of a DKMS, QKC, or link) are likewise absorbed at the next periodic tick, which rebuilds the arc list and capacity matrix from the current topology snapshot before solving.

\subsection{Scaling}

The controller solves two LPs in sequence. Phase 1 has $2|K|\cdot|E| + |K| + 1$ variables (the $x^{k}_{ij}$ per commodity and arc, one slack $\sigma_k$ per commodity, and the scalar $\lambda$); phase 2 has $2|K|\cdot|E| + |K|$ variables (the same flow re-optimization plus one bonus $\eta_k$ per commodity, with $\lambda$ substituted by its phase-1 optimum and therefore not present as a variable). Both phases share the same constraint structure of $|K|\cdot|V|$ flow-conservation equalities plus $|E|$ undirected-capacity inequalities, for a total of $|K|\cdot|V| + |E|$ constraints. Since $|K| = |V_{\text{DKMS}}|\cdot(|V_{\text{DKMS}}|-1)$, the size of the LP grows quadratically with the number of DKMSs. Table~\ref{tab:escalado} reports, for a family of random topologies of average degree $4$ (so that $|E|\approx 2|V_{\text{DKMS}}|$) generated with a fixed seed, the analytic LP size together with the RAM regime observed on EKS with \texttt{microlp} (an in-memory revised-simplex solver) under a memory limit of $8$\,GB assigned to the SDN pod; the empirical outcome is the out-of-memory threshold, which the solver first crosses at $|V_{\text{DKMS}}| = 31$.

\begin{table}[t]
\centering
\caption{Analytic LP size and observed RAM regime of the \texttt{microlp} solver on random topologies of average degree $4$ ($|E|\approx 2|V_{\text{DKMS}}|$, fixed seed) with $u_{ij} = 2\,000$ keys/s and an $8$\,GB SDN pod. Variable counts follow $2|K|\cdot|E| + |K| + 1$; the RAM column is measured up to $|V_{\text{DKMS}}| = 30$, the OOM threshold is empirical at $|V_{\text{DKMS}}| = 31$, and the $N = 40$ row is an analytic size projection (not runnable under the $8$\,GB limit, where any $|V_{\text{DKMS}}| \geq 31$ is OOM-killed).}
\label{tab:escalado}
\begin{tabular}{@{}cccccc@{}}
\toprule
$|V_{\text{DKMS}}|$ & Commodities & $|E|$ & LP variables & SDN RAM & State \\
\midrule
10  & 90     & 20 & $\sim 3\,700$    & $<1$ GB   & stable \\
20  & 380    & 40 & $\sim 31\,000$   & $<2$ GB   & stable \\
25  & 600    & 50 & $\sim 61\,000$   & $\sim 4$ GB & stable \\
30  & 870    & 60 & $\sim 105\,000$  & $\sim 8$ GB & marginal \\
31  & 930    & 62 & $\sim 116\,000$  & $>8$ GB   & OOM \\
40  & 1\,560 & 80 & $\sim 251\,000$  & ---       & projected \\
\bottomrule
\end{tabular}
\end{table}

The bottleneck is the LP's constraint matrix, not the solving operation itself: \texttt{microlp} internally reconstructs the basis and the factorization in each solve, which scales as $\mathcal{O}(|K|^2 |E|)$. For larger topologies, standard \emph{column generation} techniques over pre-computed sets of $k$-shortest paths would allow reducing variables to $\mathcal{O}(|K|\cdot k)$, at the cost of losing the \emph{splittable} multipath guarantee over arbitrary arcs. Another alternative is to replace the solver with an external one based on interior-point methods and \emph{sparse} format (HiGHS \cite{huangfu2018highs}, CLP), which in preliminary tests scales well up to $|V_{\text{DKMS}}| = 100$.

The forwarding table of a single QKC is indexed by destination, so it holds $\mathcal{O}(|V_{\text{DKMS}}| \cdot \Delta)$ entries, where $\Delta$ is the average degree of the graph (at most $\Delta$ next hops per destination after WCMP quantization); summed over all QKCs, the aggregate forwarding state of the network is $\mathcal{O}(|V_{\text{DKMS}}|^2 \cdot \Delta)$. For $|V_{\text{DKMS}}| = 100$ and $\Delta = 3$ this amounts to on the order of $30\,000$ entries across the network (a few MB in memory). It is manageable without persistence; if a QKC restarts, its table is lost and the SDN reinstalls it on the next recomputation.

\subsection{Operational correctness of the SDN}
\label{subsec:correctitud}

For the control plane to be usable in production, it is not enough that it solves an LP: it must maintain observable \emph{invariants} and degrade reasonably in the face of failures. The proposed architecture guarantees the following properties; they refer to the allocation computed and installed after a successful solve, while during the shortest-path fallback episodes the delivered rates remain those of phase~1 and the per-link load of the fallback routes is best-effort rather than LP-certified.

\paragraph{End-to-end admission} Since the LP imposes the capacity constraint \eqref{eq:cap} \emph{over the entire} graph, if the solver returns $r_k > 0$ for a commodity, by construction there exists a feasible \emph{edge-flow} distribution that delivers that flow from $s_k$ to $t_k$ respecting \emph{all} capacities. The scenario --- typical of the hop-by-hop model --- in which a commodity is admitted on the first link but is blocked at the third is not possible.

\paragraph{Sum of rates $\leq$ link capacity} Since the constraints \eqref{eq:cap} are imposed explicitly for each $\{i,j\} \in E$ summing both directions of the arc and all commodities, the sum of assigned rates crossing a link never exceeds its estimated capacity $u_{ij}$. This is what allows each QKC to limit the link's traffic with its \emph{token bucket} without the need to renegotiate per commodity. The LP's unit of account coincides with the physical one in the evaluated configuration: a $256$-bit transported key consumes exactly one $256$-bit link key per hop, because the QKC one-time-pads only the $32$-byte onion body while the routing headers travel as cleartext frame fields. Keys larger than the $32$-byte link-key granularity, or the nested onion modes ---which serialize inner layers into the payload--- would instead consume $\lceil \text{payload} / 32~\text{B} \rceil$ link keys per hop.

\paragraph{Consistency under recomputation} Each recomputation moves the network from one feasible allocation to the next. A key already in flight is processed with whatever forwarding table its current QKC holds, and the next key of the same flow sees the updated table once the SDN has pushed it; the rate snapshot and the forwarding tables converge to the new solution within one control cycle. Because both the previous and the new allocation individually satisfy the per-link capacity constraint~\eqref{eq:cap}, the transition does not require a \emph{tear-down}: it is a recompute-and-republish.

\paragraph{Fallback upon SDN outage} If the SDN becomes unavailable, the QKCs continue forwarding according to their last installed WCMP table and the DKMSs continue applying the last \texttt{rate} they received as a token-bucket ceiling. The data plane therefore keeps operating against the most recent consistent allocation, although it cannot adapt to capacity or demand changes until the SDN recovers. A finer-grained fallback based on TTL per forwarding entry, \emph{pull-on-miss} with exponential backoff at the QKC, and an explicit refresh policy at the DKMS is discussed in Section~\ref{sec:conclusions} as future work.

\paragraph{Link failure} When a link $\{i,j\}$ goes down, its capacity $u_{ij}$ collapses to zero in the capacity matrix and a new LP is solved. Commodities with alternative connectivity keep $r_k > 0$; those that become isolated are signaled to the \emph{ingress} DKMS with $r_k = 0$ and are responsible for blocking or falling back to the PQC plane. The capacity matrix is updated through the SDN's link-capacity interface; automatic detection and reporting of a downed link by the adjacent QKC (via OTP \emph{relay} failure or capacity collapse) is part of the intended operation and is not yet fully wired in the prototype.

\paragraph{Telemetry and verification} The SDN collects telemetry from the DKMSs through the demand-reporting channel itself, which is periodic and idempotent: each report overwrites the complete buffer state of the emitting DKMS. This allows an operator to continuously verify that (i)~the rate $r_k$ served to each DKMS is consistent with the solution of the last LP, (ii)~the observed effective rates respect the \emph{token buckets}, and (iii)~the actual occupancy of each link is within the solver's tolerances. The metric $\lambda$ is by itself a good summary of the network's health: a high $\lambda$ indicates buffers that fill quickly and an uncongested network; a low $\lambda$ indicates that the network is close to saturation.

\subsection{Emulation of QKD links}
\label{subsec:quditto}

To emulated QKD links, the system integrates \emph{Quditto}  \cite{app14031018}, an open-access automated emulation platform that integrates high-fidelity quantum channel modeling with a standardized key-delivery API identical to the API in a physical QKD node.  Quditto emulates pairs of QKD nodes connected by a quantum channel and continuously produces \emph{shared} key material between the two ends, which is what the QKC then uses to encrypt the \emph{relay} messages via OTP.

\paragraph{Link model} Each link $\{i,j\}$ in the emulated QKD topology corresponds to a Quditto instance configured with the attenuation parameters and reference rate of the simplified model~\eqref{eq:rate}. Internally, Quditto delivers to both ends the \emph{same} logical key buffer ($k_{\text{me}} = k_{\text{partner}}$), consistent with the LP's undirected graph hypothesis: there is a single scalar $u_{ij}$ per link, not two. This is important for the correctness of the formulation, because if the simulator produced different material per direction, the shared-capacity constraint \eqref{eq:cap} would no longer hold.

\paragraph{Buffers per peer, not per link} Although the physical link is unique, the DKMS maintains a buffer \emph{per peer} (specifically, indexed by the neighboring DKMS's identifier and the \emph{origin/destination} role), and not per physical link. This reflects that the consumption of each end is independent: the local SAE $\text{A}$ that pulls keys from DKMS $s$ to talk with $d$ empties the buffer of $s$ for $d$, but not the symmetric buffer of $d$ for $s$. This is why the commodities $(s,d)$ and $(d,s)$ are distinct in the LP, even though they share the physical capacity $u_{ij}$ of the links they traverse.

\paragraph{Link capacity in the LP} The capacity $u_{ij}$ of each link is derived from the attenuation model~\eqref{eq:rate}, $u_{ij} = R_0\,10^{-\alpha d_{ij}/10}$, and is loaded into the SDN together with the topology, so that long, lossy links carry a proportionally smaller share of the budget. The SDN additionally exposes a link-capacity interface that updates a single edge's capacity and triggers a recomputation; this is the hook through which a future feedback loop could feed an observed rate $\hat{u}_{ij}(t)$, reflecting attenuation changes or a pause in the QKD engine, without perturbing the solver manually. In the present prototype the capacities are fixed at load time from the per-link attenuation parameters.

\paragraph{Why this nuances performance} If Quditto were replaced by an \emph{ideal} key source ($u_{ij}$ constant and infinite), the resulting scalar $\lambda^\star$ from the LP would be arbitrarily large and uninformative: the network would always be uncongested. With Quditto reproducing $R_0 \, 10^{-\alpha L /10}$, longer links saturate sooner and the solver redistributes load among the available paths via edge-flow. In the experiments all links share the same deployed length of $5$~km, so capacities are uniform and the differences observed between topologies are structural rather than attenuation-driven; heterogeneous-length deployments exercise this mechanism directly and are noted as future work.

\section{Results}
\label{sec:results}

\paragraph{\textbf{Experimental methodology}}
The evaluation runs as two consecutive phases within the same simulation, on the EKS cluster \texttt{dkms1} (region \texttt{eu-north-1}). The first \emph{buffer-fill} phase initializes all ENC (encryption-side) buffers empty and switches no SAE traffic on: the MCMCF-$\lambda$ control plane is the only agent acting on the network, so it should drive every commodity toward the simultaneous-filling solution predicted by Proposition~\ref{prop:sync}. The second \emph{ETSI~014 round-trip} phase keeps the same simulation and superimposes a synchronized load injector: $30$ worker pods, each handling $500$ ordered SAE pairs (master and slave), generate ETSI~014 \texttt{enc\_keys}/\texttt{dec\_keys} round-trips at a mean per-pair rate of $\lambda_{\text{SAE}} = 2$~req/s with uniform-jittered inter-arrivals, for a total of $30\,000$ SAEs and a nominal aggregate demand of $30\,000$ round-trips/s ($60\,000$ ETSI~014 calls/s, counting the \texttt{enc\_keys} and \texttt{dec\_keys} halves separately) after a ramp of about $135$~s ($50 \to 500$ active pairs per worker in steps of $50$ every $15$~s) and a $120$~s hold.

For reference, this offered load deliberately exceeds what the cluster can sustainably deliver. Across the five realizations the network sustains a byte-matched end-to-end delivery of $7$--$10$\,k~keys/s during the hold (median $\approx 8.6$\,k~keys/s), where each completed round-trip consumes exactly one end-to-end relayed transport key. This ceiling is topology-dependent ---from $\sim\!7.1$\,k~keys/s on the sparsest Barab\'asi--Albert realization to $\sim\!10.1$\,k~keys/s on the RGG--- and lies well below the aggregate max-min-fair throughput the SDN allocates over these graphs given the link capacities ($\sim 23$--$35$\,k~keys/s), which places the binding constraint in the per-node generation-and-relay pipeline and the SDN fair-share rather than in the QKD link rate. Moreover, request latency under back-pressure throttles each SAE pair well below its nominal $2$~req/s, so the realized offered load plateaus near $11$--$12$\,k round-trips/s rather than the nominal $30$\,k; of that realized load, $62$--$85\%$ completes and the remainder is shed as ETSI~014 \texttt{429} responses ---the contract's intended back-pressure behavior.

\subsection{Buffer-fill phase: consistency with Proposition~\ref{prop:sync}}
\label{subsec:bufferfill}

Table~\ref{tab:saturation} aggregates the buffer-fill measurements. For each topology, every one of the $380$ commodities is tracked from its first \texttt{emit} event until its ENC buffer first crosses $0.95 \cdot B_{\max}$; the table reports the quartiles and maximum of those per-commodity saturation times $T_k$, together with the median of the ratio $T_k^{\text{obs}} / T_k^{\text{theor}}$, where $T_k^{\text{theor}} = B_{\max} / r_k^{\text{mm}}$ and $r_k^{\text{mm}}$ is the rate commodity $k$ receives under an ideal max--min fair allocation in which every commodity is restricted to its single shortest path and all $380$ contend for the deployed link capacities.

\begin{table}[t]
\centering
\caption{Buffer-fill phase. $T_k$ is the time, in seconds, between the first \texttt{emit} of commodity~$k$ and the first sample crossing $95\,\%$ of $B_{\max}=8\,192$ keys. Quartiles and the maximum are computed over the saturated commodities of each topology, and inherit the $5$~s sampling cadence of the generator-state trace; the ratio is the median of $T_k^{\text{obs}} / T_k^{\text{theor}}$, with $T_k^{\text{theor}}$ derived from an ideal max--min fair single-shortest-path allocation over the deployed (uniform) link capacities, with all $380$ commodities contending.}
\label{tab:saturation}
\footnotesize
\setlength{\tabcolsep}{4pt}
\begin{tabular}{@{}lcccccc@{}}
\toprule
Topology & Saturated & $Q_1$ [s] & Median [s] & $Q_3$ [s] & Max [s] & Ratio \\
\midrule
BA     & $380/380$ & $110$ & $175$ & $225$ & $325$ & $0.86$ \\
BA2    & $380/380$ & $134$ & $245$ & $325$ & $340$ & $0.75$ \\
ER     & $380/380$ & $165$ & $210$ & $220$ & $230$ & $1.41$ \\
RGG    & $380/380$ & $185$ & $205$ & $210$ & $285$ & $1.47$ \\
SECOQC & $380/380$ & $145$ & $245$ & $255$ & $270$ & $0.96$ \\
\bottomrule
\end{tabular}
\end{table}

The headline result is the first column: \textbf{all $380$ commodities saturate in all five realizations}, with median times between $175$ and $245$~s and every last buffer full before the six-minute mark. This is consistent with the synchronization contract of Proposition~\ref{prop:sync}: the LP does not merely admit every commodity, it drives all of them to their fill target within a common horizon, without any commodity being starved by the topology's bottlenecks.

The ratio column benchmarks the allocation, and three properties of the benchmark must be stated before reading it. First, it is \emph{static}: each commodity is assumed to hold its max--min rate for the entire fill, whereas in the real system a commodity that completes releases its share and the remaining ones accelerate over the tail. Second, $T_k^{\text{obs}}$ is measured up to $0.95 \cdot B_{\max}$ while $T_k^{\text{theor}}$ assumes the full buffer, shifting ratios downward by a factor of about $0.95$. Third, the observation clock starts at the first $5$~s sample with a positive emit count, which can exclude up to $10$--$20\,\%$ of the fill from the measured time. Ratios below unity are therefore expected even without any multipath contribution. With that lens, the five medians sit within $\pm 50\,\%$ of the ideal single-path benchmark ($0.75$--$1.47$), despite the system operating through a real transport plane and under the intermittent phase-2 regime of Section~\ref{subsec:lex}, which alternates multipath WCMP episodes with shortest-path fallback. BA and BA2 land lowest ($0.86$ and $0.75$), consistent with ---though not attributable solely to--- the multipath episodes routing around the hubs that throttle any single-path scheme. SECOQC sits at $0.96$, essentially at the benchmark, as expected of a ring-with-chords topology with few alternative routes. ER and RGG land above unity ($1.41$ and $1.47$): these are the two densest realizations, where the benchmark rates are highest and the relay and fallback overheads weigh correspondingly more.

\begin{figure*}[t]
    \centering
    \subcaptionbox{Erdős--Rényi ($N=20$).}{\includegraphics[width=0.48\textwidth]{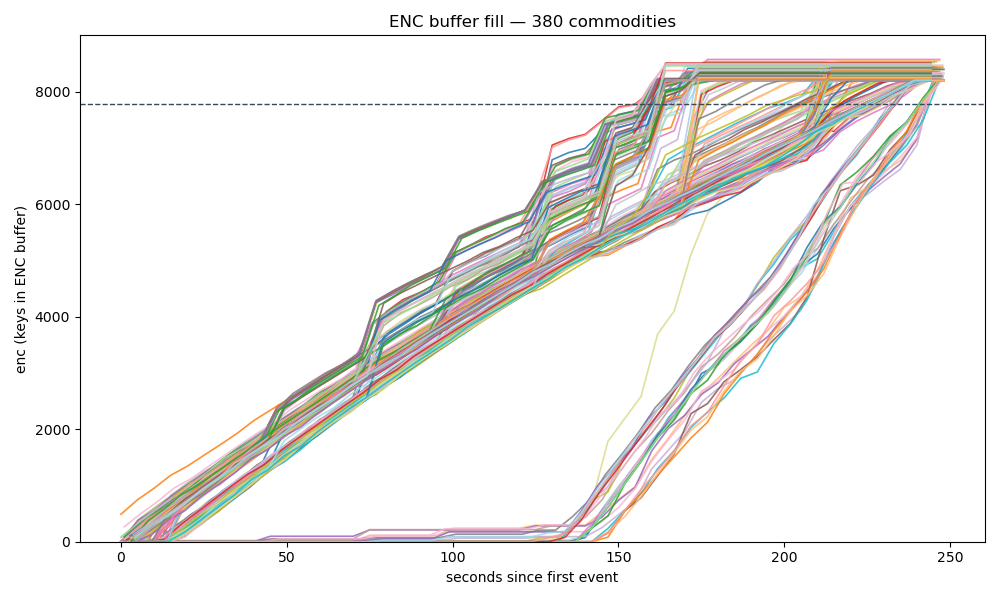}}
    \hfill
    \subcaptionbox{Random Geometric Graph ($N=20$).}{\includegraphics[width=0.48\textwidth]{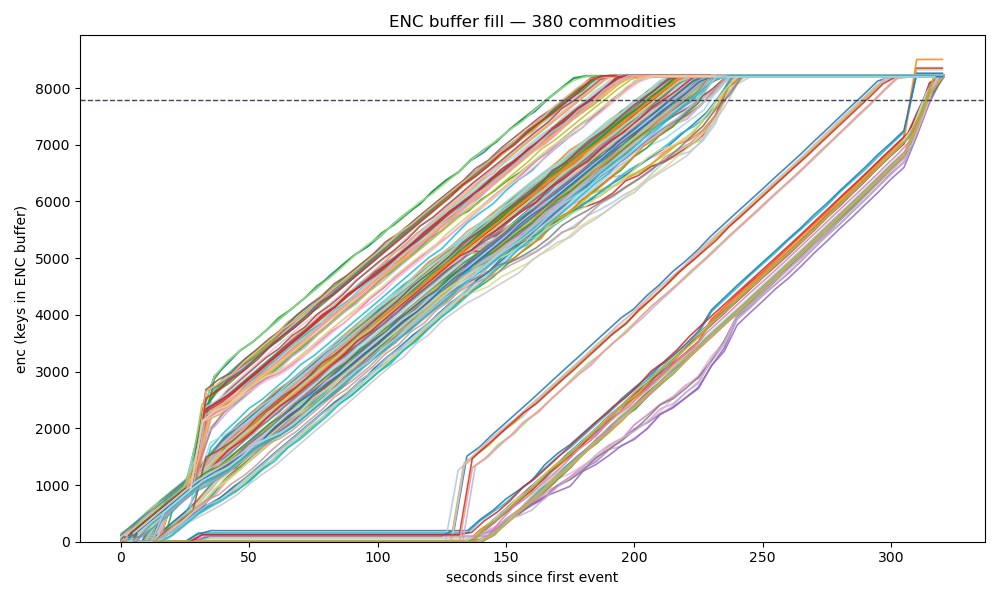}}
    \caption{ENC level per commodity over time during the buffer-fill phase; the dashed line at $0.95 \cdot B_{\max}$ marks the saturation threshold. The $380$ trajectories converge toward the same horizon rather than each pair filling independently, illustrating Proposition~\ref{prop:sync} empirically; a few trajectories briefly overshoot $B_{\max}$, which the control plane treats as a fill target rather than a hard cap.}
    \label{fig:buffer_fill}
\end{figure*}

The two panels of Figure~\ref{fig:buffer_fill} show the per-commodity ENC trajectories for ER and RGG and make the synchronization visible: trajectories converge toward a common horizon rather than each pair filling independently.

\subsection{ETSI~014 round-trip phase: SAE-level correctness and throughput}

Table~\ref{tab:etsi014} summarizes the outcome of every SAE round-trip attempted during the load phase, per topology realization.

\begin{table}[t]
\centering
\caption{ETSI~014 round-trip phase ($30$ workers $\times$ $500$ SAE pairs, $\lambda_{\text{SAE}} = 2$~req/s per pair, $N=20$). A ``match'' is a round-trip whose \texttt{enc\_keys} and \texttt{dec\_keys} halves both succeeded and returned byte-identical key material. HTTP~$429$ is the token-bucket back-pressure of Section~\ref{sec:buffers}; HTTP~$502$ is a transient nginx-ingress upstream timeout.}
\label{tab:etsi014}
\footnotesize
\setlength{\tabcolsep}{3pt}
\begin{tabular}{@{}lrrrrr@{}}
\toprule
 & BA & BA2 & ER & RGG & SECOQC \\
\midrule
Round-trips & $2\,206\,336$ & $2\,189\,544$ & $2\,279\,510$ & $2\,329\,499$ & $2\,186\,046$ \\
Match & $1\,829\,594$ & $1\,702\,655$ & $1\,896\,878$ & $2\,121\,672$ & $1\,978\,376$ \\
Match \% & $82.92$ & $77.76$ & $83.21$ & $91.08$ & $90.50$ \\
Mismatch & $0$ & $0$ & $100\,709$ & $38$ & $0$ \\
HTTP 429 & $376\,737$ & $486\,865$ & $281\,725$ & $207\,733$ & $207\,650$ \\
429 \% & $17.08$ & $22.24$ & $12.36$ & $8.92$ & $9.50$ \\
HTTP 502 & $5$ & $24$ & $24$ & $55$ & $20$ \\
Other & $0$ & $0$ & $174$ & $1$ & $0$ \\
\bottomrule
\end{tabular}
\end{table}

Three things deserve emphasis. First, \textbf{byte-level key match is perfect in three of the five realizations}: BA, BA2 and SECOQC complete millions of round-trips without a single mismatch, and RGG shows $38$ mismatches in $2.1$~million completions ($99.998\,\%$ byte-identical). The delivery pipeline ---session keys wrapped by one-time pad with transport keys that were themselves relayed under per-link OTP and the ORR's end-to-end ML-KEM wrap--- preserves key integrity at scale. The exception is ER, whose $100\,709$ mismatches ($4.4\,\%$ of round-trips) are concentrated in a burst and remain under investigation; the anomaly does not reproduce in the other four realizations and is registered as an open issue.

Second, the dominant non-success outcome everywhere is \textbf{HTTP~$429$ back-pressure, not error}: between $8.92\,\%$ (RGG) and $22.24\,\%$ (BA2) of requests are cleanly rejected by the per-SAE token buckets of Section~\ref{sec:buffers} because the offered load exceeds what the key-generation plane can sustain. Genuine transport errors are negligible: HTTP~$502$ (transient ingress timeouts) peaks at $55$ per run and the residual ``Other'' count at $174$ (ER), out of $2.2$--$2.3$~million attempts; Figure~\ref{fig:sae_outcome} shows the RGG breakdown. RGG's position at the low end of the back-pressure range is consistent with its link count: at uniform per-link capacity, its $46$ links give it roughly $50$--$65\,\%$ more aggregate key-generation capacity than the three sparsest realizations ($28$--$31$ links) for the same offered load.

Third, the system \textbf{degrades exactly as designed}: the correct queueing behavior under a nominal offered load several times above the rate the cluster can sustainably generate is to convert the excess into explicit, retryable back-pressure while keeping every served key byte-correct ---which the system achieves outright in three realizations and on $99.998\,\%$ of completions in RGG, ER's anomaly excepted. This is precisely what the $429$/match split of Table~\ref{tab:etsi014} shows.

\begin{figure}[t]
    \centering
    \includegraphics[width=\columnwidth]{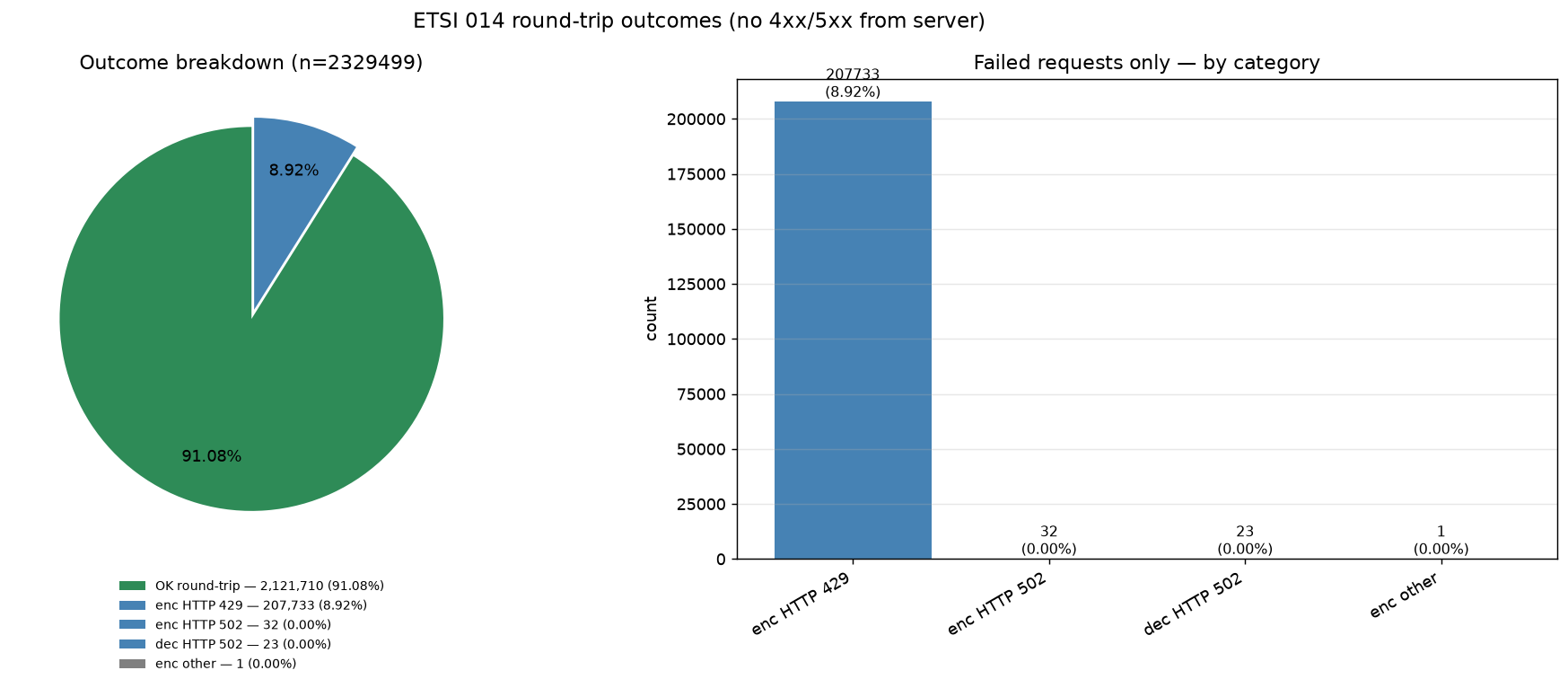}
    \caption{Outcome breakdown of the $2\,329\,499$ SAE round-trips of the RGG run. Green: byte-identical match; orange: token-bucket back-pressure (HTTP~$429$); the residual slice aggregates $55$ transient HTTP~$502$, $38$ mismatches and $1$ client error.}
    \label{fig:sae_outcome}
\end{figure}

Figure~\ref{fig:sae_rps} shows the RGG round-trip rate over time. The staircase ramp is visible up to $\approx 135$~s ($50 \to 500$ pairs per worker, reaching $\approx 12\,000$~rps with the match line tracking the total); the periodic dips to zero coincide with the injector's $15$~s ramp steps, during which each worker briefly halts traffic to register its next batch of SAE pairs before resuming injection, driving the instantaneous rate momentarily to zero. After the ramp, the regime holds at about $11$--$12$~krps of total round-trips ($\approx 9.6$~krps byte-matched) for the entire hold window, with the gap between the two lines being precisely the token-bucket rejections analyzed below. This byte-matched rate ($\approx 9.6$~krps) transiently exceeds the sustained generation ceiling of $\sim 6.6$\,k~keys/s: over a finite hold, matches also drain the pre-filled transport-key buffers ---an aggregate stock of $380 \times 8\,192 \approx 3.1$~million keys, which comfortably covers the $\approx 3$\,k~keys/s deficit over the $120$~s window--- so the $6.6$\,k figure bounds the indefinitely sustainable rate, not the transient one.

\begin{figure}[t]
    \centering
    \includegraphics[width=\columnwidth]{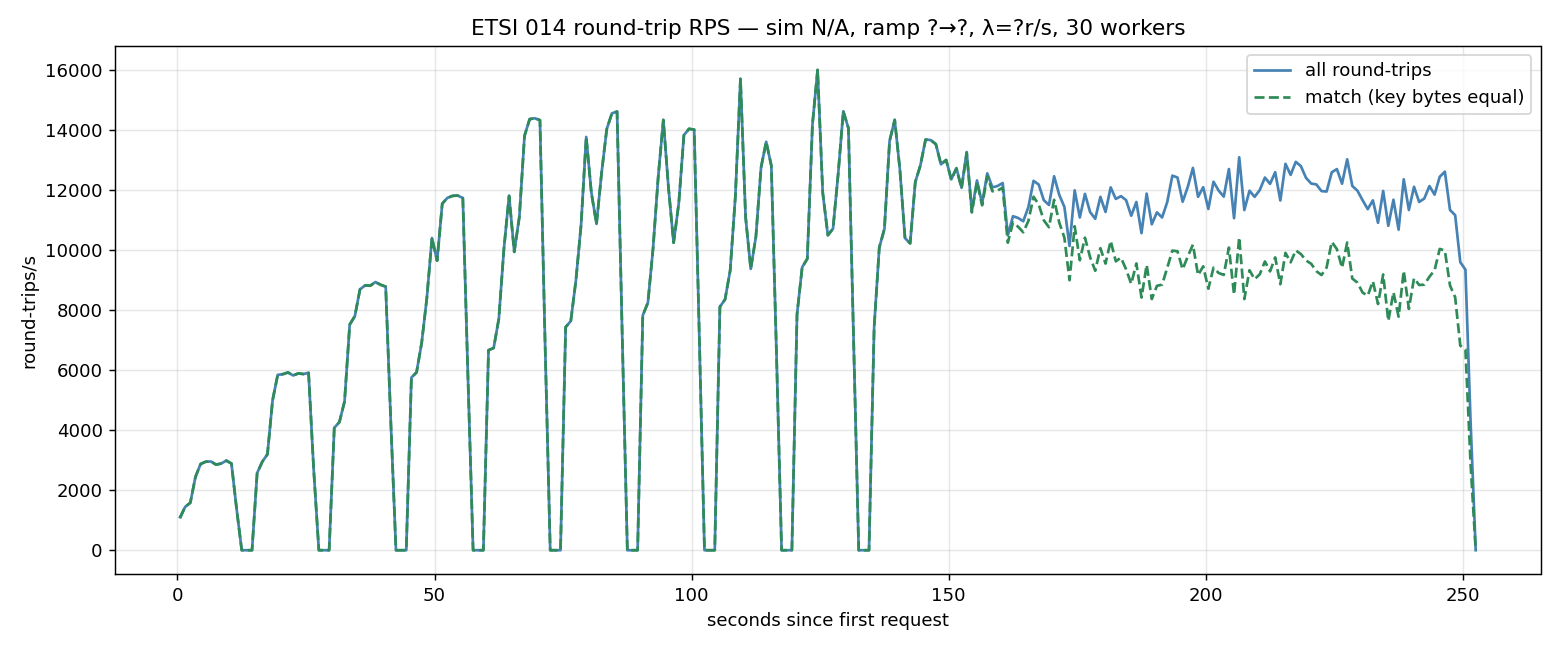}
    \caption{ETSI~014 round-trip rate over time (RGG, $N=20$). The ramp ($50 \to 500$ active pairs per worker) is visible up to $\approx 135$~s; the hold sits around $11$--$12$~krps total ($\approx 9.6$~krps byte-matched).}
    \label{fig:sae_rps}
\end{figure}

Figure~\ref{fig:sae_match_vs_429} decomposes the same run into its two dominant outcomes. The match rate (green) and the $429$ rate (red) are complementary: when the ramp pushes the offered load past the sustainable rate, the token buckets convert the excess into rejections while the match rate stays pinned at the transiently deliverable rate. The red band is not noise: it is the visible, quantitative signature of the graceful-degradation design of the system.

\begin{figure}[t]
    \centering
    \includegraphics[width=\columnwidth]{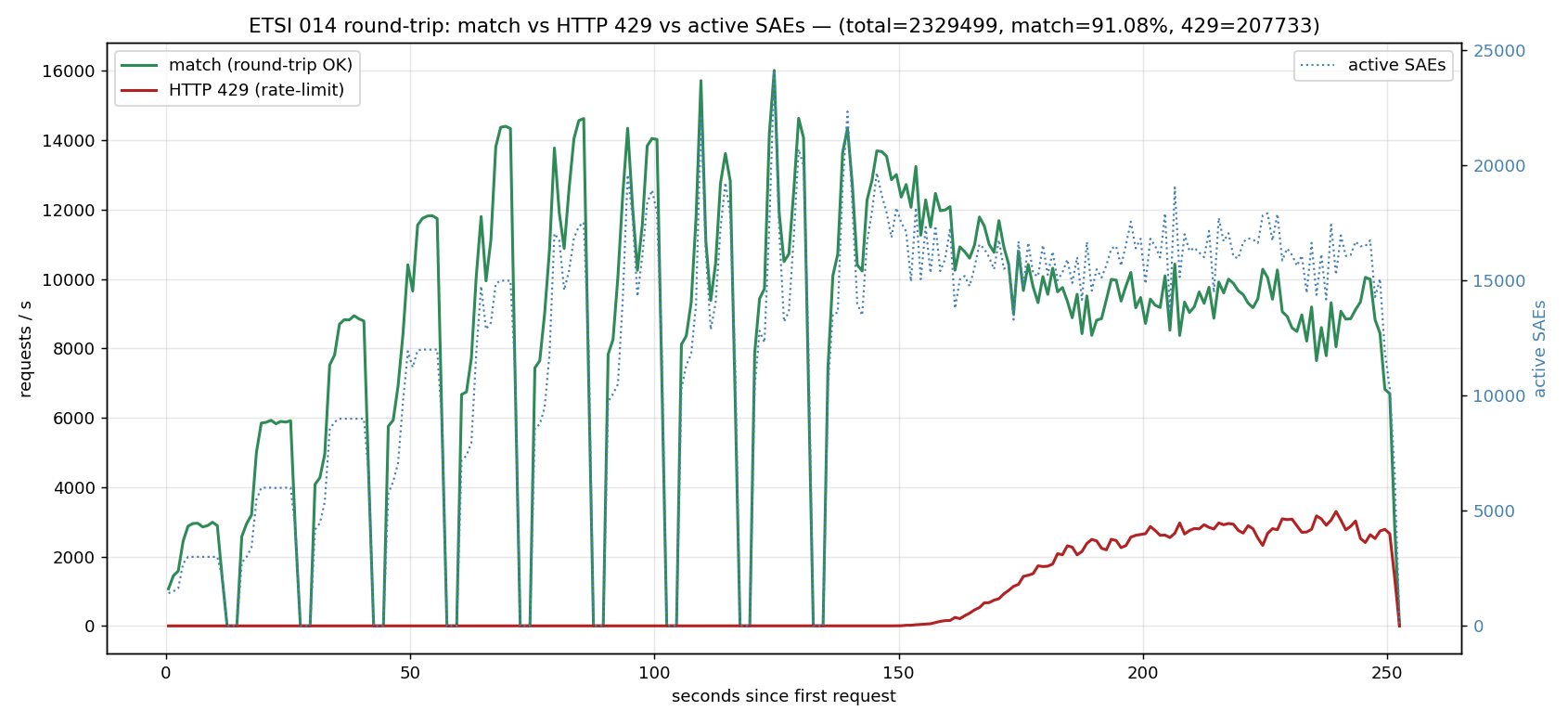}
    \caption{Byte-matched round-trips (green) versus token-bucket rejections (red) per second, ETSI~014 round-trip phase, RGG. The two series are complementary: the excess offered load is absorbed as HTTP~$429$ back-pressure while the served keys remain byte-correct.}
    \label{fig:sae_match_vs_429}
\end{figure}

Figure~\ref{fig:sae_latency} examines the latency of the two ETSI~014 halves. Median (p50) latency stays low throughout the run for both \texttt{enc\_keys} and \texttt{dec\_keys}; the p95 tail spikes transiently during the final ramp steps and then settles at $2$--$3$~s in the back-pressured region; this tail is dominated by \emph{accepted} requests waiting on the synchronous ETSI~020 delivery to the peer KMS under load, and on the shared ingress, while token-bucket rejections themselves fail fast. The figure confirms that median latency is unaffected by back-pressure; only the p95/p99 tails respond.

\begin{figure}[t]
    \centering
    \includegraphics[width=\columnwidth]{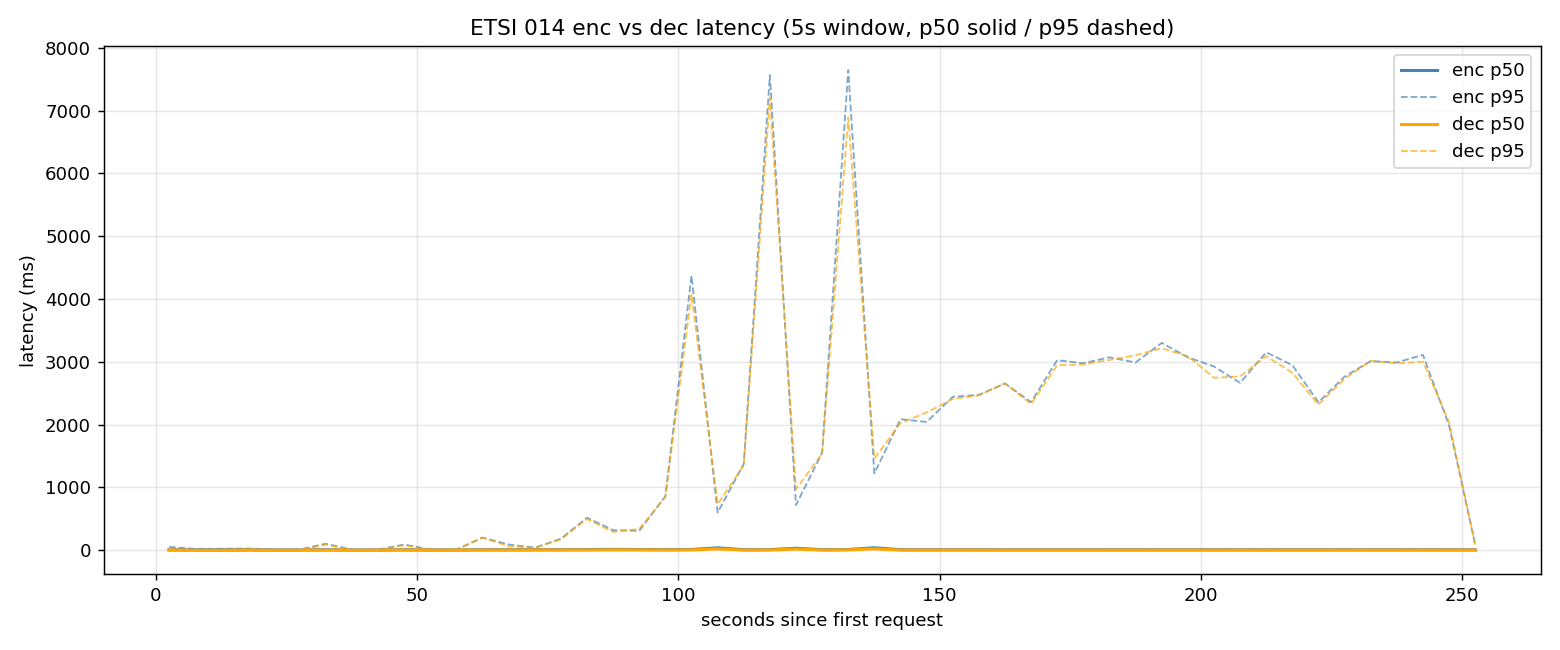}
    \caption{\texttt{enc\_keys} and \texttt{dec\_keys} latency per $5$-second window (RGG). Median (p50) latency stays below $50$~ms throughout the run; the p95 tail spikes transiently during the final ramp steps ($t \approx 100$--$135$~s) and then settles at $2$--$3$~s in the back-pressured region, dominated by accepted requests waiting on the synchronous ETSI~020 delivery to the peer KMS under load (token-bucket rejections fail fast).}
    \label{fig:sae_latency}
\end{figure}

Figure~\ref{fig:sae_keys_per_dkms} verifies that the work is distributed essentially uniformly across the $20$ DKMSs of the topology (the spread from the least to the most loaded DKMS is about $7\,\%$), the expected outcome of a uniform offered load combined with the globally fair SDN allocation, with the per-SAE buckets equalizing shares within each DKMS.

\begin{figure}[t]
    \centering
    \includegraphics[width=\columnwidth]{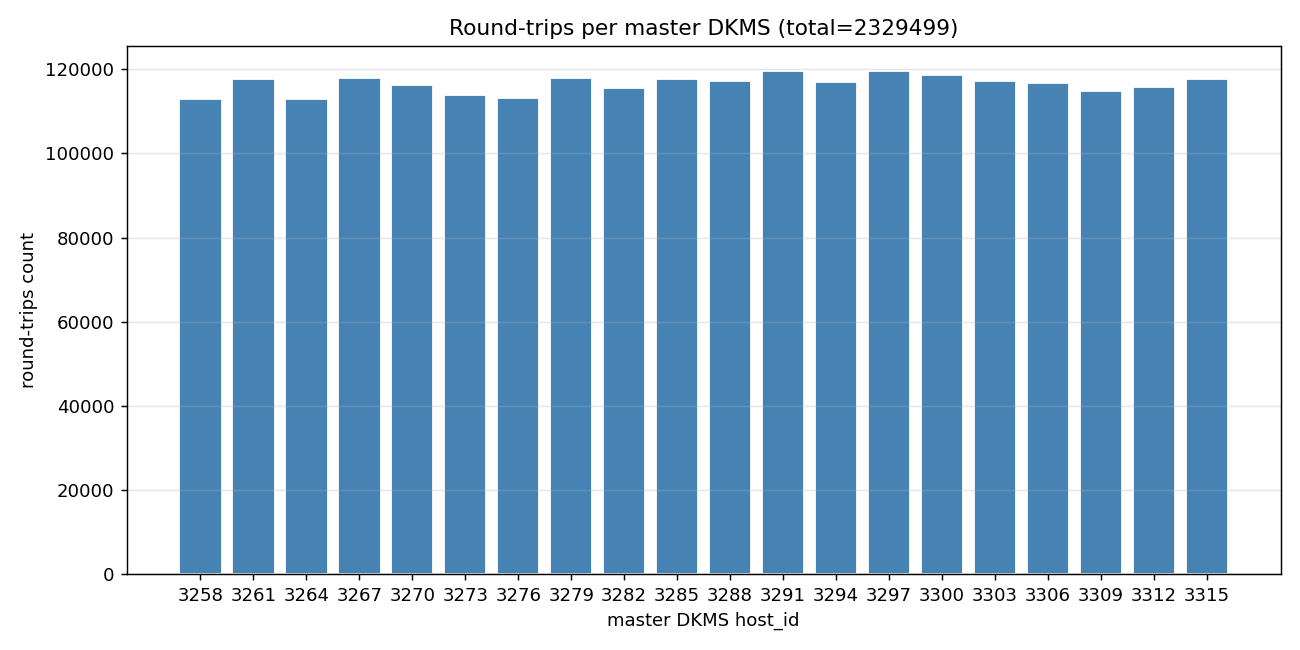}
    \caption{Round-trips served per master DKMS (RGG, $N=20$). The distribution is essentially flat across the $20$ DKMSs (range $112\,000$--$120\,000$).}
    \label{fig:sae_keys_per_dkms}
\end{figure}

A practical implication of BA versus BA2 deserves a brief comment. The two graphs are drawn from the \emph{same} family (Barabási--Albert, $N=20$; $31$ vs.\ $28$ links) and differ only in the random seed used to generate the link set, yet their match percentages differ by more than five points ($82.92\,\%$ vs.\ $77.76\,\%$) and their back-pressure rates equally so ($17.08\,\%$ vs.\ $22.24\,\%$ of HTTP~$429$); Figure~\ref{fig:sae_match_vs_429_ba2_compare} shows the BA2 time series, whose back-pressure region is visibly larger than RGG's in Figure~\ref{fig:sae_match_vs_429}. The structural lesson is that a single realization of a random family is not enough to characterize the family: a topology-level percentage should be quoted with a confidence interval over several seeds.

\begin{figure}[t]
    \centering
    \includegraphics[width=\columnwidth]{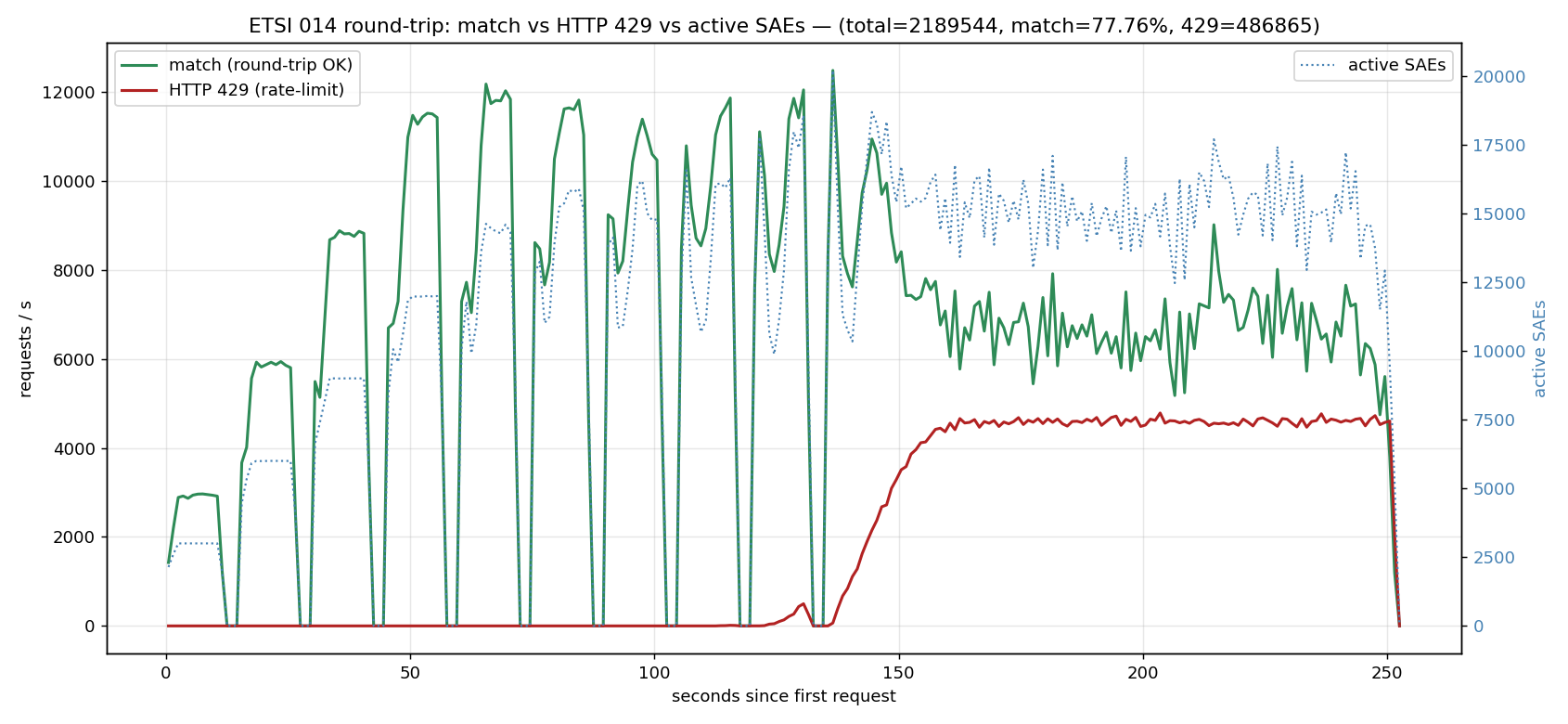}
    \caption{Same axes as Figure~\ref{fig:sae_match_vs_429} but for BA2 (a different seed of the same Barabási--Albert family, $N=20$). The back-pressure region (red) is substantially larger than in the RGG run ($22.24\,\%$ of requests vs.\ $8.92\,\%$) and five points above its sibling seed BA ($17.08\,\%$): a single seed does not characterize the family.}
    \label{fig:sae_match_vs_429_ba2_compare}
\end{figure}

\subsection{Comparative synthesis}

Table~\ref{tab:resumen} synthesizes the dominant behavior of each realization across both phases.

\begin{table}[t]
\centering
\caption{Comparative synthesis of the five topology realizations ($N=20$, $380$ commodities). The ratio row repeats the buffer-fill benchmark of Table~\ref{tab:saturation}: values below $1$ mean faster than an ideal max--min fair single-path allocation, subject to the downward measurement biases discussed in the text.}
\label{tab:resumen}
\footnotesize
\setlength{\tabcolsep}{3pt}
\begin{tabular}{@{}lccccc@{}}
\toprule
 & BA & BA2 & ER & RGG & SECOQC \\
\midrule
Saturated commodities & $380$ & $380$ & $380$ & $380$ & $380$ \\
Median fill time [s] & $175$ & $245$ & $210$ & $205$ & $245$ \\
Fill-time ratio vs.\ single path & $0.86$ & $0.75$ & $1.41$ & $1.47$ & $0.96$ \\
Match \% & $82.92$ & $77.76$ & $83.21$ & $91.08$ & $90.50$ \\
Back-pressure (429) \% & $17.08$ & $22.24$ & $12.36$ & $8.92$ & $9.50$ \\
Mismatches & $0$ & $0$ & $100\,709$ & $38$ & $0$ \\
\bottomrule
\end{tabular}
\end{table}

Read together, the two phases support three claims. (i)~The MCMCF-$\lambda$ control plane delivers on its synchronization contract in every topology family evaluated: $380/380$ commodities saturated, all within six minutes, regardless of whether the graph is hub-dominated, sparse-random, geometric or ring-like. (ii)~Benchmarked against an ideal max--min single-path allocation, the measured fill times of all five realizations sit within $\pm 50\,\%$ of the ideal (median ratios $0.75$--$1.47$), with the sub-unity ratios of the hub-dominated BA seeds ($0.86$ and $0.75$) consistent with ---though not attributable solely to--- the multipath episodes routing around their hubs, and the two densest graphs (ER, RGG) trailing the benchmark under the intermittent phase-2 regime of Section~\ref{subsec:lex}. (iii)~Under an offered load several times above what the cluster can sustainably generate (a nominal $30\,000$ round-trips/s against a sustained generation ceiling of $\sim 6.6$\,k~keys/s, roughly $4.5\times$), the system absorbs the excess as HTTP~$429$ back-pressure rather than as data loss: byte-level key match is perfect in three of the five realizations and holds on $99.998\,\%$ of RGG completions, with ER's mismatch burst standing as the sole open anomaly.

\section{Conclusions and Future Work}
\label{sec:conclusions}

This work has presented the design of a distributed quantum key manager (DKMS) whose central contribution is to move the capacity allocation of the network to a global optimization problem: an SDN control plane that formulates the filling of all the key buffers as an instance of Maximum Concurrent Multi-Commodity Flow with demand scaling (MCMCF-$\lambda$), shifting the capacity and route decision from a hop-by-hop reservation scheme ---which can admit a flow on the first link only for it to die further along--- to a centralized control plane that guarantees end-to-end admission over the complete graph for every allocation its solver completes. The filling-synchronization property gives the operator a single scalar, $\lambda$, that summarizes the saturation state of the entire network.

Several lines remain open. The immediate ones are operational: replacing the in-memory simplex solver with a sparse interior-point backend (HiGHS) so that the lexicographic second phase ---today succeeding in only about a third of the solves at scale--- delivers its multipath forwarding and residual-bandwidth bonus consistently instead of intermittently, and hardening the QKC-side failure detection so that a downed link is reported to the SDN without operator intervention. A second group concerns the control plane itself: enriching the forwarding scheme with per-entry TTLs and \emph{pull-on-miss} refresh, feeding observed link rates $\hat{u}_{ij}(t)$ back into the capacity matrix, and studying the stability of the control loop under adversarial demand patterns. A third group concerns evaluation: investigating the byte-level mismatch observed under sustained ER load, whose burst concentration suggests a race in the delivery pipeline rather than a systematic defect of the measurement harness; deploying campaigns with heterogeneous link lengths, so that the attenuation-aware allocation 
is exercised on non-uniform capacities.

\bibliographystyle{IEEEtran}
\bibliography{lib}

@Article{app14031018,
AUTHOR = {Martin, Raul and Lopez, Blanca and Vidal, Ivan and Valera, Francisco and Nogales, Borja},
TITLE = {Service for Deploying Digital Twins of QKD Networks},
JOURNAL = {Applied Sciences},
VOLUME = {14},
YEAR = {2024},
NUMBER = {3},
ARTICLE-NUMBER = {1018},
URL = {https://www.mdpi.com/2076-3417/14/3/1018},
ISSN = {2076-3417},
DOI = {10.3390/app14031018}
}

@article{10.1145/3730575,
author = {Dervisevic, Emir and Tankovic, Amina and Fazel, Ehsan and Kompella, Ramana and Fazio, Peppino and Voznak, Miroslav and Mehic, Miralem},
title = {Quantum Key Distribution Networks - Key Management: A Survey},
year = {2025},
issue_date = {October 2025},
publisher = {Association for Computing Machinery},
address = {New York, NY, USA},
volume = {57},
number = {10},
issn = {0360-0300},
url = {https://doi.org/10.1145/3730575},
doi = {10.1145/3730575},
journal = {ACM Comput. Surv.},
month = may,
articleno = {257},
numpages = {36}
}

@techreport{mceliece1978,
  author      = {McEliece, Robert J.},
  title       = {A Public-Key Cryptosystem Based on Algebraic Coding Theory},
  institution = {Jet Propulsion Laboratory},
  note        = {DSN Progress Report 42--44},
  pages       = {114--116},
  year        = {1978}
}

@inproceedings{bb84,
  author    = {Bennett, Charles H. and Brassard, Gilles},
  title     = {Quantum Cryptography: Public Key Distribution and Coin Tossing},
  booktitle = {Proc. IEEE Int. Conf. on Computers, Systems and Signal Processing},
  pages     = {175--179},
  year      = {1984}
}

@inproceedings{elliott2005darpa,
  author    = {Elliott, Chip and Colvin, Alexander and Pearson, David and
               Pikalo, Oleksiy and Schlafer, John and Yeh, Henry},
  title     = {Current Status of the {DARPA} Quantum Network},
  booktitle = {Quantum Information and Computation III},
  series    = {Proc. SPIE},
  volume    = {5815},
  pages     = {138--149},
  year      = {2005},
  doi       = {10.1117/12.606489}
}

@article{poppe2008secoqc,
  author  = {Poppe, Andreas and Peev, Momtchil and Maurhart, Oliver},
  title   = {Outline of the {SECOQC} Quantum-Key-Distribution Network in Vienna},
  journal = {International Journal of Quantum Information},
  volume  = {6},
  number  = {2},
  pages   = {209--218},
  year    = {2008},
  doi     = {10.1142/S0219749908003529}
}

@article{sasaki2011tokyo,
  author  = {Sasaki, M. and Fujiwara, M. and Ishizuka, H. and Klaus, W. and
             Wakui, K. and Takeoka, M. and Miki, S. and Yamashita, T. and
             Wang, Z. and Tanaka, A. and others},
  title   = {Field Test of Quantum Key Distribution in the {Tokyo} {QKD} Network},
  journal = {Optics Express},
  volume  = {19},
  number  = {11},
  pages   = {10387--10409},
  year    = {2011},
  doi     = {10.1364/OE.19.010387}
}

@article{stucki2011swissquantum,
  author  = {Stucki, D. and Legr{\'e}, M. and Buntschu, F. and Clausen, B. and
             Felber, N. and Gisin, N. and Henzen, L. and Junod, P. and
             Litzistorf, G. and Monbaron, P. and others},
  title   = {Long-term Performance of the {SwissQuantum} Quantum Key Distribution
             Network in a Field Environment},
  journal = {New Journal of Physics},
  volume  = {13},
  number  = {12},
  pages   = {123001},
  year    = {2011},
  doi     = {10.1088/1367-2630/13/12/123001}
}

@article{dynes2019,
  author  = {Dynes, J. F. and Wonfor, A. and Tam, W. W.-S. and Sharpe, A. W. and
             Takahashi, R. and Lucamarini, M. and Plews, A. and Yuan, Z. L. and
             Dixon, A. R. and Cho, J. and Tanizawa, Y. and Elbers, J.-P. and
             Grei{\ss}er, H. and White, I. H. and Penty, R. V. and Shields, A. J.},
  title   = {Cambridge quantum network},
  journal = {npj Quantum Information},
  volume  = {5},
  number  = {1},
  pages   = {101},
  year    = {2019},
  doi     = {10.1038/s41534-019-0221-4}
}

@article{chen2021,
  author  = {Chen, Yu-Ao and Zhang, Qiang and Chen, Teng-Yun and Cai, Wen-Qi and
             Liao, Sheng-Kai and Zhang, Jun and Chen, Kai and Yin, Juan and
             Ren, Ji-Gang and Chen, Zhu and others},
  title   = {An integrated space-to-ground quantum communication network over
             4,600 kilometres},
  journal = {Nature},
  volume  = {589},
  number  = {7841},
  pages   = {214--219},
  year    = {2021},
  doi     = {10.1038/s41586-020-03093-8}
}

@article{martin2024madqci,
  author  = {Martin, V. and Brito, J. P. and Ort{\'i}z, L. and M{\'e}ndez, R. B. and
             Buruaga, J. S. and Vicente, R. J. and Sebasti{\'a}n-Lombra{\~n}a, A. and
             Rinc{\'o}n, D. and P{\'e}rez, F. and S{\'a}nchez, C. and others},
  title   = {{MadQCI}: A Heterogeneous and Scalable {SDN-QKD} Network Deployed in
             Production Facilities},
  journal = {npj Quantum Information},
  volume  = {10},
  number  = {1},
  pages   = {80},
  year    = {2024},
  doi     = {10.1038/s41534-024-00873-2}
}

@article{aguado2019sdn,
  author  = {Aguado, Alejandro and L{\'o}pez, V{\'i}ctor and L{\'o}pez, Diego and
             Peev, Momtchil and Poppe, Andreas and Pastor, Antonio and
             Folgueira, Jes{\'u}s and Mart{\'i}n, Vicente},
  title   = {The Engineering of Software-Defined Quantum Key Distribution Networks},
  journal = {IEEE Communications Magazine},
  volume  = {57},
  number  = {7},
  pages   = {20--26},
  year    = {2019},
  doi     = {10.1109/MCOM.2019.1800763}
}

@misc{openqkd,
  author       = {{OpenQKD Consortium}},
  title        = {{OPENQKD}: Open European Quantum Key Distribution Testbed},
  howpublished = {H2020 Project 857156, CORDIS, European Commission},
  year         = {2019},
  note         = {\url{https://cordis.europa.eu/project/id/857156}}
}

@misc{euroqci,
  author       = {{European Commission}},
  title        = {{European Quantum Communication Infrastructure (EuroQCI)}},
  year         = {2025},
  url          = {https://digital-strategy.ec.europa.eu/en/policies/european-quantum-communication-infrastructure-euroqci},
  note         = {Official European Commission page, accessed 2026-03-27}
}

@article{cao2022qinternet,
  author  = {Cao, Yuan and Zhao, Yongli and Wang, Qin and Zhang, Jie and
             Ng, Soon Xin and Hanzo, Lajos},
  title   = {The Evolution of Quantum Key Distribution Networks: On the Road to
             the {Qinternet}},
  journal = {IEEE Communications Surveys \& Tutorials},
  volume  = {24},
  number  = {2},
  pages   = {839--894},
  year    = {2022},
  doi     = {10.1109/COMST.2022.3144219}
}

@article{mehic2020networking,
  author  = {Mehic, Miralem and Niemiec, Marcin and Rass, Stefan and Ma, Jiajun and
             Peev, Momtchil and Aguado, Alejandro and Martin, Vicente and
             Schauer, Stefan and Poppe, Andreas and Pacher, Christoph and
             Voznak, Miroslav},
  title   = {Quantum Key Distribution: A Networking Perspective},
  journal = {ACM Computing Surveys},
  volume  = {53},
  number  = {5},
  pages   = {1--41},
  year    = {2020},
  doi     = {10.1145/3402192}
}

@article{pirandola2017plob,
  author  = {Pirandola, Stefano and Laurenza, Riccardo and Ottaviani, Carlo and
             Banchi, Leonardo},
  title   = {Fundamental Limits of Repeaterless Quantum Communications},
  journal = {Nature Communications},
  volume  = {8},
  pages   = {15043},
  year    = {2017},
  doi     = {10.1038/ncomms15043}
}

@article{yuan2018mbps,
  author  = {Yuan, Z. and Plews, A. and Takahashi, R. and Doi, K. and Tam, W. and
             Sharpe, A. W. and Dixon, A. R. and Lavelle, E. and Dynes, J. F. and
             Murakami, A. and Kujiraoka, M. and Lucamarini, M. and Tanizawa, Y. and
             Sato, H. and Shields, A. J.},
  title   = {10-{Mb/s} Quantum Key Distribution},
  journal = {Journal of Lightwave Technology},
  volume  = {36},
  number  = {16},
  pages   = {3427--3433},
  year    = {2018},
  doi     = {10.1109/JLT.2018.2843136}
}

@article{liu2023tfqkd1000,
  author  = {Liu, Yang and Zhang, Wei-Jun and Jiang, Cong and Chen, Jiu-Peng and
             Zhang, Chi and Pan, Wen-Xin and Ma, Di and Dong, Hao and
             Xiong, Jia-Min and Zhang, Cheng-Jun and Li, Hao and Wang, Rui-Chun and
             Wu, Jun and Chen, Teng-Yun and You, Lixing and Wang, Xiang-Bin and
             Zhang, Qiang and Pan, Jian-Wei},
  title   = {Experimental Twin-Field Quantum Key Distribution over 1000 km Fiber
             Distance},
  journal = {Physical Review Letters},
  volume  = {130},
  number  = {21},
  pages   = {210801},
  year    = {2023},
  doi     = {10.1103/PhysRevLett.130.210801}
}

@article{quantumrepeater,
  author  = {Azuma, Koji and Economou, Sophia E. and Elkouss, David and
             Hilaire, Paul and Jiang, Liang and Lo, Hoi-Kwong and Tzitrin, Ilan},
  title   = {Quantum Repeaters: From Quantum Networks to the Quantum Internet},
  journal = {Reviews of Modern Physics},
  volume  = {95},
  number  = {4},
  pages   = {045006},
  year    = {2023},
  doi     = {10.1103/RevModPhys.95.045006}
}

@article{tysowski2018,
  author  = {Tysowski, Piotr K. and Ling, Xinhua and L{\"u}tkenhaus, Norbert and
             Mosca, Michele},
  title   = {The Engineering of a Scalable Multi-Site Communications System
             Utilizing Quantum Key Distribution ({QKD})},
  journal = {Quantum Science and Technology},
  volume  = {3},
  number  = {2},
  pages   = {024001},
  year    = {2018},
  doi     = {10.1088/2058-9565/aa9a5d}
}

@techreport{etsi014,
  author       = {{ETSI}},
  title        = {{Quantum Key Distribution (QKD); Protocol and Data Format of REST-based Key Delivery API}},
  institution  = {{European Telecommunications Standards Institute}},
  number       = {GS QKD 014 V1.1.1},
  year         = {2019},
  month        = feb,
  url          = {https://www.etsi.org/deliver/etsi_gs/QKD/001_099/014/01.01.01_60/gs_qkd014v010101p.pdf}
}

@techreport{etsi020,
  title = {ETSI GS QKD 020: Quantum Key Distribution (QKD); Interoperable Key Management System (KMS) API},
  institution  = {{European Telecommunications Standards Institute}},
  author = {{ETSI}},
  year = {2023},
url = {https://www.etsi.org/deliver/etsi_gs/QKD/001_099/020/01.01.01_60/gs_QKD020v010101p.pdf}
}

@techreport{etsi015,
  author       = {{ETSI}},
  title        = {{Quantum Key Distribution (QKD); Control Interface for Software Defined Networks}},
  institution  = {{European Telecommunications Standards Institute}},
  number       = {GS QKD 015 V2.1.1},
  year         = {2022},
  month        = apr,
  url          = {https://www.etsi.org/deliver/etsi_gs/QKD/001_099/015/02.01.01_60/gs_qkd015v020101p.pdf}
}

@misc{etsi018,
  author       = {{ETSI}},
  title        = {Quantum Key Distribution (QKD); Orchestration Interface for
                  Software Defined Networks},
  howpublished = {ETSI GS QKD 018 V1.1.1},
  year         = {2022},
  note         = {European Telecommunications Standards Institute}
}

@techreport{itu_y3800,
  author       = {{ITU-T}},
  title        = {{Recommendation ITU-T Y.3800: Overview on Networks Supporting Quantum Key Distribution}},
  institution  = {{International Telecommunication Union}},
  number       = {Y.3800},
  year         = {2019},
  month        = oct,
  url          = {https://www.itu.int/rec/T-REC-Y.3800-201910-I/en},
  note         = {Approved 2019-10-25}
}

@article{yang2018secretkey,
  author  = {Yang, Chao and Zhang, Hongqi and Su, Jinhai},
  title   = {Quantum Key Distribution Network: Optimal Secret-Key-Aware Routing
             Method for Trust Relaying},
  journal = {China Communications},
  volume  = {15},
  number  = {2},
  pages   = {33--45},
  year    = {2018},
  doi     = {10.1109/CC.2018.8300270}
}

@article{yu2022skpcr,
  author  = {Yu, Xiaosong and Liu, Yuhang and Zou, Xiangyu and Cao, Yuan and
             Zhao, Yongli and Nag, Avishek and Zhang, Jie},
  title   = {Secret-Key Provisioning with Collaborative Routing in
             Partially-Trusted-Relay-based Quantum-Key-Distribution-Secured
             Optical Networks},
  journal = {Journal of Lightwave Technology},
  volume  = {40},
  number  = {12},
  pages   = {3530--3542},
  year    = {2022},
  doi     = {10.1109/JLT.2022.3153992}
}

@inproceedings{amer2020routing,
  author    = {Amer, Omar and Krawec, Walter O. and Wang, Bing},
  title     = {Efficient Routing for Quantum Key Distribution Networks},
  booktitle = {2020 IEEE International Conference on Quantum Computing and
               Engineering (QCE)},
  pages     = {137--147},
  year      = {2020},
  doi       = {10.1109/QCE49297.2020.00027}
}

@article{zhang2023icqkd,
  author  = {Zhang, Qiaolun and Ayoub, Omran and Wu, Jun and Lin, Xi and
             Tornatore, Massimo},
  title   = {{IC-QKD}: An Information-Centric Quantum Key Distribution Network},
  journal = {IEEE Communications Magazine},
  volume  = {61},
  number  = {12},
  pages   = {148--154},
  year    = {2023},
  doi     = {10.1109/MCOM.001.2200520}
}

@misc{grillo2021,
  author       = {Grillo, Milo and Dowhuszko, Alexis A. and
                  Khalighi, Mohammad-Ali and H{\"a}m{\"a}l{\"a}inen, Jyri},
  title        = {Resource Allocation in a Quantum Key Distribution Network with
                  {LEO} and {GEO} Trusted-Repeaters},
  howpublished = {arXiv:2107.07595},
  year         = {2021},
  note         = {\url{https://arxiv.org/abs/2107.07595}}
}

@article{zeng2024hybrid,
  author  = {Zeng, Pei and Zhou, Debayan and Lu, Weijie and Sinclair, Nathan and
             others},
  title   = {Practical Hybrid {PQC-QKD} Protocols with Enhanced Security and
             Performance},
  journal = {arXiv:2411.01086},
  year    = {2024},
  note    = {\url{https://arxiv.org/abs/2411.01086}}
}

@inproceedings{braun2023hybrid,
  author    = {Braun, Ralf-Peter and Geitz, Marc and Doll, Florian and others},
  title     = {Hybrid {QKD} \& {PQC} Protocols Implemented in the {Berlin}
               {OpenQKD} Testbed},
  booktitle = {2023 8th International Conference on Frontiers of Signal Processing
               (ICFSP)},
  year      = {2023},
  doi       = {10.1109/ICFSP59665.2023.10372894}
}

@article{shahrokhi1990,
  author  = {Shahrokhi, Farhad and Matula, David W.},
  title   = {{The Maximum Concurrent Flow Problem}},
  journal = {Journal of the ACM},
  volume  = {37},
  number  = {2},
  pages   = {318--334},
  year    = {1990},
  doi     = {10.1145/77600.77620},
  url     = {https://doi.org/10.1145/77600.77620}
}

@book{ahuja1993,
  author    = {Ahuja, Ravindra K. and Magnanti, Thomas L. and Orlin, James B.},
  title     = {{Network Flows: Theory, Algorithms, and Applications}},
  publisher = {Prentice Hall},
  address   = {Englewood Cliffs, NJ},
  year      = {1993},
  isbn      = {978-0136175490}
}

@article{huangfu2018highs,
  author  = {Huangfu, Qi and Hall, J. A. J.},
  title   = {{Parallelizing the Dual Revised Simplex Method}},
  journal = {Mathematical Programming Computation},
  volume  = {10},
  number  = {1},
  pages   = {119--142},
  year    = {2018},
  doi     = {10.1007/s12532-017-0130-5},
  url     = {https://doi.org/10.1007/s12532-017-0130-5}
}

\vfill

\end{document}